\documentclass[journal]{IEEEtran}
\usepackage{lettrine}
\usepackage[cmex10]{amsmath}
\usepackage{amssymb}
\usepackage{amsthm}
\usepackage{cite}
\usepackage{amsmath}
\usepackage{amssymb}
\usepackage{booktabs}
\usepackage{balance}
\usepackage{color}
\usepackage{subfigure}
\usepackage{ulem}
\usepackage{multicol}
\usepackage{graphicx}
\usepackage{epsfig}
\usepackage{epstopdf}
\usepackage{cases}
\usepackage{algorithmic}
\usepackage[ruled,vlined]{algorithm2e}
\usepackage{bm}
\usepackage{amsmath, amssymb, amsthm}
\usepackage{upgreek}
\usepackage{CJK}
\usepackage{multirow}
\usepackage{float}
\usepackage{stfloats}

\theoremstyle{definition} 
\theoremstyle{definition} 
\theoremstyle{plain} \newtheorem{lemma}{Lemma}
\theoremstyle{plain} \newtheorem{proposition}{Proposition}

\usepackage{mathrsfs}
\usepackage{threeparttable}
\usepackage{ifpdf}
\ifCLASSINFOpdf
\else
\fi
\begin{document}
%
\title{On Throughput Optimization and Bound Analysis in Cache-Enabled Fiber-Wireless Networks}
\author{\IEEEauthorblockN{Zhuojia Gu, Hancheng Lu,~\IEEEmembership{Senior~Member,~IEEE}, Zuqing Zhu,~\IEEEmembership{Senior~Member,~IEEE}}

\thanks{Z. Gu, H. Lu and Z. Zhu are with Key Laboratory of Wireless-Optical Communications, Chinese Academy of Sciences, School of Information Science and Technology, University of Science and Technology of China, Hefei 230027, China (Email: guzj@mail.ustc.edu.cn; hclu@ustc.edu.cn; zqzhu@ieee.org).}
\thanks{This work was supported in part by the National Science Foundation of China (No.61771445, No.61631017, No.91538203). This paper was presented in part at IEEE GLOBECOM 2018 \cite{Joint-power-allocation-and-caching}.}}

\maketitle

%


\begin{abstract}
With the dense deployment of millimeter wave (mmWave) front ends and popularization of bandwidth-intensive applications, shared backhaul in fiber-wireless (FiWi) networks is still facing a bandwidth crunch. To alleviate the backhaul pressure, in this paper, caching capability is enabled at the edge of FiWi networks, i.e., optical network unit access points (ONU-APs). On the other hand, as both power budget and backhaul bandwidth in FiWi networks are constrained, it is challenging to properly leverage power for caching and that for wireless transmission to achieve superior system performance. As caching has a significant impact on resource allocation, we reconsider performance optimization and analysis in cache-enabled FiWi networks. Firstly, in the cache-enabled FiWi network with mmWave, we formulate the joint power allocation and caching problem, with the goal to maximize the downlink throughput. A two-stage algorithm is then proposed to solve the problem.
Secondly, to investigate the theoretical capacity of the cache-enabled FiWi network with mmWave, we derive an upper bound of the downlink throughput by analyzing properties of the average rate of wireless links. Particularly, we show that appropriate power allocation for wireless transmission and caching at ONU-APs is essential to achieve higher throughput. The numerical and simulation results validate our theoretical analysis and demonstrate the proposed algorithm can approach the analytical upper bound.
\end{abstract}

\hspace{8mm}
\par

%
\IEEEpeerreviewmaketitle

\section{Introduction}
\lettrine[lines=2]{N}{owadays}, the Internet is foreseen to suffer from a severe bandwidth crunch in the coming years \cite{Joint-power-allocation-and-caching, Enhancing-end-user-bandwidth, Lu2015_Network}. According to Cisco's forecast \cite{ref1}, global Internet traffic will achieve a three-fold increase over the next 5 years, and video applications will be approximately 82\% of all Internet traffic by 2022, up from 75\% in 2017. Evidently, supporting such increasing demand for Internet traffic is then vital for future access networks. To cope with the growing bandwidth crunch in access networks, fiber-wireless (FiWi) networks, as an integration of fiber backhaul and wireless access, have emerged as a promising solution to provide an efficient ``last mile'' Internet access and have been widely deployed, especially in the field of multimedia communications \cite{ref2}. In light of the complementary technical merits of fiber backhaul and wireless access \cite{Gong2014_JLT, Liang2017_JOCN}, FiWi networks have attracted intensive research interests over the past two decades, which attempt to combine high capacity and reliability of fiber backhaul with high flexibility and ubiquitous connectivity of wireless front ends \cite{high-capacity-and-flexibility}.

To alleviate the bandwidth pressure at wireless front ends, the use of large-bandwidth at millimeter wave (mmWave) frequency bands, between 30 and 300 GHz, becomes a promising solution for fifth generation (5G) cellular networks and has attracted considerable attention recently \cite{Coverage-and-rate-analysis, Millimeter-wave-mobile}. In this paper, we consider a FiWi network with mmWave, where mmWave frequency bands are used for the wireless front ends. In this network, the fiber backhaul bandwidth is shared by all mmWave front ends. To cope with blockage effects of mmWave links, the mmWave front ends are supposed to be densely deployed to provide seamless coverage\cite{Content-placement-in-cache-enabled}. With the popularization of bandwidth-intensive applications (i.e., augmented or virtual reality, etc.), shared backhaul in the FiWi network is still facing a bandwidth crunch.

Some research efforts have been devoted to next-generation passive optical networks (NG-PONs) \cite{NGPON2-technology-and-standards,Emerging-optical-access-network}, which aim at increasing the capacity of state-of-the-art PON based fiber backhaul. However, existing optical network units (ONUs) and optical line terminals (OLTs) are not directly compatible with the technical specifications of NG-PONs, so ONUs and OLTs must be replaced or upgraded to support NG-PONs, which is a costly method at present.
Fortunately, it has been reported that a large number of files in wireless networks are repeatedly requested, such as high volume video files and map data in vehicular networks \cite{Self-sustaining-caching}.
Based on this fact, we enable caching capability at the edge of FiWi networks to alleviate the backhaul pressure. However, in this case, caching has a significant impact on resource allocation, and the performance of cache-enabled FiWi networks should be reconsidered.

\subsection{Related Work}
\textit{1) Caching in optical fiber networks}:
There have been growing recent researches towards data caching as an approach for alleviating the bandwidth pressure in optical fiber networks. In \cite{Enhancing-end-user-bandwidth}, an architecture consisting of an ONU with an associated storage unit was proposed to save traffic in the feeder fiber and improve the system throughput as well as mean packet delays. In \cite{ref12}, a framework of software-defined PONs with caches was introduced to achieve a substantial increase in served video requests. In \cite{ref13}, a dynamic bandwidth algorithm based on local storage video-on-demand delivery in PONs was proposed, and the achievable throughput levels have been improved when a local storage is used to assist video-on-demand delivery. However, most of them mainly focus on providing storage capacity in PONs to improve the network throughput, which did not consider the impact of wireless front ends on resource allocation, so these methods cannot be directly applied to FiWi networks.

\textit{2) Caching in wireless networks}: Data caching has also been used to offload backhaul traffic in wireless networks. In \cite{Content-placement-in-cache-enabled}, a cache-enabled mmWave small cell network was proposed to achieve the maximum successful content delivery probability. The cache-related coverage probability of a mmWave network was derived in \cite{Cache-enabled-hetnets}, with the aid of stochastic geometry.  In \cite{An-analysis-on-caching-placement}, a two-tier heterogeneous network consisting of both sub-6 GHz and mmWave BSs with caches was proposed, where the authors considered the average success
probability (ASP) of file delivery as the performance metric and proposed efficient schemes to maximize the ASP file delivery. The work in \cite{Cache-aided-millimeter-wave} investigated the performance of cache-enabled mmWave ad hoc networks and demonstrated the effectiveness of off-line caching with respect to blockages.

\textit{3) Power consumption in cache-enabled networks}:
Power consumption has become one of the main concerns in both optical and wireless communication networks due to the rapid growth in Internet traffic, especially when involving caches in networks, which introduces additional power consumption that cannot be ignored. The power efficiency of different caching hardware technology has been shown in \cite{caching-efficiency}. As an example, the commonly used storage medium high-speed solid state disk (SSD) consumes about 5 W (Watt) of caching power per 100 GB.
In \cite{Energy-Efficiency-of-an}, the authors study the impact of deploying video cache servers on the power consumption over hybrid optical switching, and demonstrate a careful dimensioning of cache size to achieve high power efficiency. In \cite{Energy-efficiency-of}, the authors studied whether local caching at base stations can improve the energy efficiency by taking into account the impact of caching power. In \cite{On-Energy-Efficient-Edge}, an optimal caching scheme was proposed to minimize the energy consumption in cache-aided heterogeneous networks.

\textit{4) Mobile edge computing in FiWi networks}:
In recent years, many researches have applied FiWi networks to mobile edge computing (MEC). MEC-enabled FiWi networks were first introduced in \cite{Cloudlet-enhanced-fiber-wireless} to reduce the packet delay by leveraging the computing resources of MEC. In \cite{Collaborative-computation-offloading}, FiWi networks were adopted to provide supports for the coexistence of centralized cloud and MEC, and two collaborative computation offloading schemes were proposed to achieve superior computation offloading performance. Meanwhile, the impact of dynamic bandwidth allocation in MEC-enabled FiWi networks was studied in \cite{Mobile-edge-computing-versus}, and a novel unified resource management scheme incorporating both centralized cloud and MEC computation offloading activities was proposed. In \cite{FiWi-enhanced-vehicular-edge, Task-offloading-in-vehicular}, MEC-enabled FiWi networks were introduced in vehicular networks to support collaborative computation task offloading of intelligent connected vehicles. A FiWi enhanced smart-grid communication architecture was proposed in \cite{Big-data-acquisition-under} and the problem of data acquisition under failures in FiWi enhanced smart-grids was investigated. Existing research work on MEC-enabled FiWi networks mainly focused on the computation offloading capability.

\subsection{Motivation and Contribution}
The current studies on the edge of FiWi networks mainly focus on providing computation capability at the network edge to enhance the network performance. However, there is a lack of studies focusing on the impact of caching capability at the edge of FiWi networks, which plays an important role in improving the throughput of FiWi networks.
Note that in FiWi networks, both the wireless front ends and the shared fiber backhaul by all wireless front ends will make an impact on the system throughput, which is different from either optical fiber networks or wireless networks.
Although existing studies have put many efforts on the impact of caching in optical fiber networks or wireless networks, it is non-trivial to achieve optimal performance in cache-enabled FiWi networks due to the significant impact of caching on resource allocation.
Firstly, as fiber backhaul is shared by all wireless front ends, resource allocation at ONU-APs is coupled. Secondly, power allocation and caching should be jointly optimized at ONU-APs. In cache-enabled FiWi networks, power consumed for caching depending on the memory size of caching files is non-neglected. Therefore, there is a tradeoff between power allocated for caching and that for wireless transmission at ONU-APs when downlink throughput is considered. If more power is used for wireless transmission, higher downlink throughput can be achieved. In this case, with constrained power budget at ONU-APs, less power will be allocated for caching, putting more bandwidth pressure on fiber backhaul.

To combat above challenges in cache-enabled FiWi networks, we conduct performance optimization and analysis with consideration of the significant impact of caching on resource allocation. MmWave based wireless front ends are assumed in this paper. The main contributions are summarized as follows.
\begin{itemize}
  \item In a FiWi network with mmWave, we propose to equip caches at ONU-APs to alleviate the backhaul pressure. In the proposed cache-enabled FiWi network where both the backhaul bandwidth and power budget at ONU-APs are assumed to be constrained, we formulate the joint power allocation and caching optimization problem as a mixed integer nonlinear programming (MINLP) problem, with the goal to maximize the downlink throughput.
  \item To solve the joint optimization problem, we propose a two-stage algorithm. At the first stage, we propose a volume adjustable backhaul constrained water-filling method (VABWF) using convex optimization to derive the expression of optimal wireless transmission power allocation at each ONU-AP. At the second stage, we reformulate the problem as a multiple-choice knapsack problem (MCKP) by exploiting the highest-popularity-first property of files. Then, a dynamic programming algorithm based on the proposed VABWF method is proposed to approach the optimal power allocation and caching strategies.
  \item To investigate the theoretical capacity of the proposed cache-enabled FiWi network with mmWave, we first derive a tractable expression of the average rate of mmWave links and find that increasing the average rate does not necessarily improve the downlink throughput. Then, we provide a theoretical upper bound of the downlink throughput by utilizing the analytical properties of the average rate. Simulation results demonstrate that the proposed algorithm can approach the theoretical downlink throughput upper bound.
\end{itemize}

The rest of this paper is organized as follows. We first introduce the system model in Section \uppercase\expandafter{\romannumeral2}. In Section \uppercase\expandafter{\romannumeral3}, we describe the proposed two-stage algorithm. In Section \uppercase\expandafter{\romannumeral4}, we give theoretical analysis on the downlink throughput upper bound by exploiting the analytical properties of the average rate of mmWave links. Numerical and simulation results are shown in Section \uppercase\expandafter{\romannumeral5} and concluding remarks are provided in Section \uppercase\expandafter{\romannumeral6}.

\section{System Model}
\subsection{Network Model}
The network model we propose is shown in Fig. \ref{fig:FiWi network architecture with caches at ONU-APs}. The FiWi network with mmWave is divided into fiber backhaul and multiple wireless networks based on mmWave front ends.
We adopt Passive Optical Network (PON) as fiber backhaul, which can leverage wavelength-division multiplexing and flexible-grid spectrum assignment \cite{Zhu2013_JLT, Gong2013_JOCN, Yin2013_JOCN} to provide optical access.
The optical line terminal (OLT) is located at the central office, connecting to a passive optical splitter through the feeder fiber. The passive optical splitter is connected to multiple ONUs through the distribution fiber. In the FiWi network, each ONU is collocated with an AP, and the integration of the ONU and AP is called integrated ONU-AP.
The ONU-APs are denoted by an index set $\mathcal{N} = \{1,\cdots, N \} $, with each ONU-AP $n \in \mathcal{N}$.
We assume that the backhaul constraint on the feeder fiber connecting the passive optical splitter and the OLT is $C$.

At the wireless front ends, considering the advantage of high bandwidth of mmWave communication, it has a tendency to become an important technology in the future 5G network \cite{Millimeter-wave-mobile, Coverage-and-rate-analysis}, so we consider mmWave based wireless front ends in this paper. ONU-APs are assumed to be spatially deployed according to a hexagonal cell planning. For easy analysis, we approximate the hexagonal cell shape as a circle with coverage radius $D$. User equipments (UEs) are distributed in the coverage of ONU-APs.
The UEs are denoted by an index set $\mathcal{K} = \{1,\cdots, K \} $, with each UE $k \in \mathcal{K}$.
We assume that each UE is associated with the closest ONU-AP, thus we denote the set of UEs associated with ONU-AP$_n$ by $\Phi_n,n \in \mathcal{N}$.

We assume that directional beamforming is adopted at each ONU-AP, and the main-lobe gain of using directional beamforming is denoted by $G$. Note that mmWave transmissions are highly sensitive to blockage, so we adopt a two-state statistical blockage model for each link as in \cite{Coverage-and-rate-analysis}, such that the probability of the link to be LOS or NLOS is a function of the distance between a UE and its serving ONU-AP. Assume that the distance between them is $r$, then the probability that a link of length $r$ is LOS or NLOS can be modeled as
    $\rho_\mathrm{L}(r) = \mathrm{e}^{-\beta r},\ \rho_\mathrm{N}(r) = 1 - \mathrm{e}^{-\beta r}$,
respectively, where $\beta$ is the decay rate depending on the blockage parameter and density \cite{Analysis-of-blockage}. Independent Nakagami fading is assumed for each link. Parameters of Nakagami fading $N_\mathrm{L}$ and $N_{\mathrm{N}}$ are assumed for LOS and NLOS links, respectively. Besides, the existing literatures have confirmed that mmWave transmissions tend to be noise-limited and interference is weak \cite{Coverage-in-heterogeneous, Tractable-model-for-rate}. Therefore, when UE$_k$ requests files from its associated ONU-AP$_n$, the received signal-to-noise ratio is given by
\begin{equation}\label{SINR-m}
  \Upsilon_{nk} = \frac{P_{nk} h_{nk} G r^{-\alpha_{m}} }{ \sigma^2  },
\end{equation}
where $P_{nk}$ is the transmit power allocated to UE$_k$ by ONU-AP$_n$, $h_{nk}$ is the Nakagami channel fading which follows Gamma distribution. The path loss exponent $\alpha_{m} = \alpha_{\mathrm{L}}$ when it is a LOS link and $\alpha_{m} = \alpha_{\mathrm{N}}$ when it is an NLOS link. $\sigma^2$ is the noise power of a mmWave link.
Then the channel state information (CSI) from ONU-AP$_n$ to UE$_k$ is  $g_{nk} = r_{nk}^{-\alpha_m} h_{nk}$.
The wireless transmission rate from ONU-AP$_n$ to UE$_k$ is given by the Shannon's theorem
\begin{equation}\label{shannon}
  R_{nk}=B\log_{2}\left(1 + \Upsilon_{nk}\right),
\end{equation}
where $B$ is the subchannel bandwidth for each UE.

\begin{figure}
\centering
\includegraphics [width=3.4in, height=2.6in]{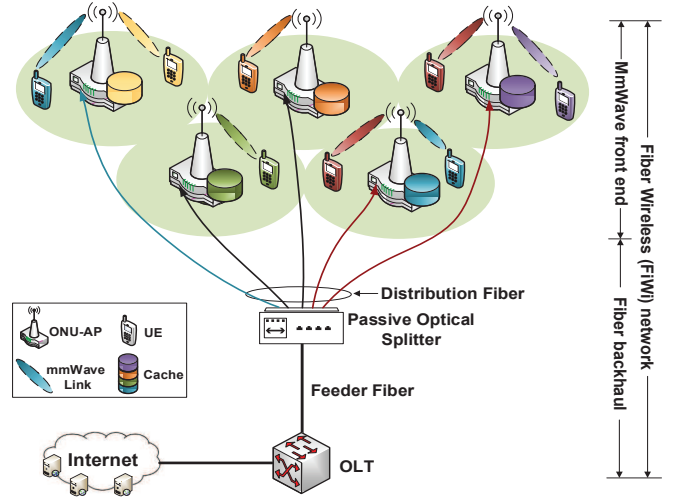}
\caption{FiWi network architecture with caches at ONU-APs.}
\label{fig:FiWi network architecture with caches at ONU-APs}
\end{figure}

\begin{table*}
  \centering
  \caption{Key Notations}
  \label{summary of the notation used in this paper}
  \begin{tabular}{|p{1.2cm}|p{6.7cm}|p{1.5cm}|p{6.6cm}|}
    \hline
    \textbf{Symbol} & \textbf{Description}  &  \textbf{Symbol} & \textbf{Description}  \\\hline
    $\mathcal{N}$ & Set of ONU-APs    &  $\Phi_{n}$ & Set of UEs associated to ONU-AP$_n$     \\\hline
    $\mathcal{J}$ & Set of files  &  $|\Phi_{n}|$ & Number of UEs associated to ONU-AP$_n$   \\\hline
    $\mathcal{K}$ & Set of UEs  &     $P_{M}$ & The preset maximum power consumed by an ONU-AP\\\hline
    $B$ &  Subchannel bandwidth for each UE &   $\alpha_\mathrm{L}$ / $\alpha_{\mathrm{N}}$  &  Path loss exponent for mmWave LOS/NLOS links   \\\hline
    $C$ &  Fiber backhaul capacity  &  $N_\mathrm{L}$ / $N_{\mathrm{N}}$  &  Parameters of Nakagami fading for LOS/NLOS links \\\hline
    $D$ &  Coverage radius of ONU-APs &   $p_j$ &  Probability of the $j$-th file being requested   \\\hline
    $G$ &  Mainlobe gain of directional beamforming at ONU-APs &  $s$ &  Size of a file \\\hline
    $Q$ &  Cache size of an ONU-AP  &    $\beta$  &  Blockage density of mmWave links    \\\hline
    $P_{nk}$ & Transmission power allocated to UE$_k$ by ONU-AP$_n$ &   $\delta$ &  Parameter of Zipf distribution   \\\hline
    $R_{nk}$ &  Transmission rate from ONU-AP$_n$ to UE$_k$  &  $\lambda$  & Density of UEs\\\hline
    $g_{nk}$ &  Channel state information from ONU-AP$_n$ to UE$_k$ &  $\rho$ &  Transmission power coefficient at ONU-APs  \\\hline
    $h_{nk}$ &  Channel fading from ONU-AP$_n$ to UE$_k$  &   $\tau$  &  The average rate of mmWave links \\\hline
    $r_{nk}$ &  Distance between ONU-AP$_n$ and UE$_k$   &    $\omega$ & Caching power coefficient \\\hline
    $x_{nj}$ &  Caching decision for the $j$-th file to ONU-AP$_n$ &  $\sigma^2$ &  Background noise power level   \\\hline
    $\Upsilon_{nk}$  &  Received signal-to-noise ratio from ONU-AP$_n$ to UE$_k$  &   $p_n^{hit}$ / $p_n^{miss}$  &  Cache hit ratio / cache miss ratio at ONU-AP$_n$\\\hline
  \end{tabular}
\end{table*}

\subsection{ONU-AP Caching Model}
Each ONU-AP is cache-enabled, and the cache size is denoted by $Q$.
The requested files are denoted by an index set $\mathcal{J} = \{1,\cdots, J \} $, with each file $j \in \mathcal{J}$.
We consider a proactive caching model, where popular files are pre-cached in ONU-APs during off-peak intervals (e.g., at night). If ONU-AP$_n$ already caches a requested file, it will respond to the request directly. Otherwise, the file should be fetched from the Internet via capacity-limited fiber backhaul. The size of each file is assumed to be $s$ bits\footnote{For analytical simplicity, the size of each file is assumed to be equal. If the file sizes are unequal, similar to \cite{FemtoCaching-Wireless, An-analysis-on-caching-placement}, a fine-grained packetization can be considered to divide each file into chunks of equal size, with each chunk being treated as an individual file, so the same analysis can still be applied.}. According to \cite{Web-caching-and}, the file popularity distribution typically follows a Zipf distribution, which is widely adopted in the literature such as \cite{Edge-caching-in-dense, Caching-placement-in-stochastic, Energy-efficiency-of}, so in this paper we model the file popularity distribution of $J$ files as Zipf distribution. Then the probability of the $j$-th ranked file being requested by UEs is $p_{j}=\frac{j^{-\delta}}{\sum_{n=1}^{J} n^{-\delta}}$,
where $\delta$ is a shape parameter that shapes the skewness of the popularity distribution. Note that the file popularity distribution is not necessarily confined to the Zipf distribution but can accommodate any discrete file popularity distribution\footnote{When the file popularity does not follow the Zipf distribution, we can index the files according to the popularity, i.e., $p_j \geq p_{j+1}, j=1, \cdots, J-1$, and the following analysis and optimization still hold since they do not rely on the assumption that file popularity follows the Zipf distribution. }. Files requested by UEs that are not cached at ONU-APs should be fetched from the Internet via fiber backhaul. Given the cache decision on the $j$-th file at ONU-AP$_n$ as $x_{nj},x \in \{0,1\}$, the expected fiber backhaul rate of uncached files for UEs associated with ONU-AP$_n$ can be written as $\sum_{k \in \Phi_n} \sum_{j \in \mathcal{J}} p_{j} (1 - x_{nj}) R_{nk}$.

\subsection{Power Consumption Model}
The total power consumption at ONU-AP$_n$ can be obtained by leveraging the typical power consumption model of wireless networks\cite{Energy-efficiency-of} and including the power used for caching as
\begin{equation}\label{Pntotal}
  P_n^{total}=\sum\limits_{k \in \Phi{n}}\rho P_{nk}+P_n^{ca}+P_n^{cc},
\end{equation}
where $P_n^{ca}$ denotes the power consumed at ONU-AP$_n$ for caching. $P_n^{cc}$ denotes the power consumed by circuits, similar to \cite{Energy-efficient-caching-for, Full-duplex-massive}, which is assumed to be a constant related to the circuit design.
$\rho$ is a coefficient that measures the impacts of power amplifier, power supply and cooling.
We use an energy-proportional model for caching power consumption \cite{On-Energy-Efficient-Edge,power-proportional-model2}, which has been widely adopted in content-centric networking for efficient use of caching power. In this model, the consumption of caching power is proportional to the total number of bits cached at an ONU-AP, i.e., $P_n^{ca}=\omega \Omega_n$, where $\Omega_n$ denotes the total number of bits cached at ONU-AP$_n$, and $\omega$ is a power coefficient that is related to the caching hardware and reflects the caching power efficiency in watt/bit. In this work, we consider the common caching device, high-speed SSD, for caching files at ONU-APs. The value of $\omega$ for SSD is $6.25 \times 10^{-12}$ watt/bit \cite{Energy-efficiency-of, In-network-caching-effect}.
Assume that the maximum power consumed by an ONU-AP can be preset, which is denoted by $P_M$. Each ONU-AP can adjust the transmission power and the caching power to achieve higher throughput without exceeding the preset maximum power limit, namely, $P_n^{total}\leq P_M$.
Here, for notational convenience, the circuit power $P_n^{cc}$ is omitted since it is assumed to be a constant.
The key notations are summarized in Table \uppercase\expandafter{\romannumeral1}.

\section{Joint Power Allocation and Caching} \label{Section:joint-power-allocation-and-caching}
In this section, we first formulate the joint power allocation and caching problem. Then, we propose an optimal transmission power allocation method, named VABWF method, as the first stage of the proposed algorithm. To perform the joint optimization of power allocation and caching, an efficient algorithm based on dynamic programming is provided by converting the problem into a standard MCKP as the second stage of the proposed algorithm.

\subsection{Problem Formulation}
Our objective is to maximize the downlink wireless access throughput while keeping the total power consumption not exceeding the preset maximum power at ONU-APs by making joint caching and power allocation decisions. This is formulated as follows,

\begin{subequations}
\begin{align}
\textbf{P1: } \max_{\bm{x},\bm{P}} \label{objective function}   &\quad\sum_{n\in \mathcal{N}}\sum_{k \in \Phi_{n}} R_{nk} \\
\text{s.t.}    &\quad \sum_{k \in \Phi_{n}}\rho P_{nk}+\omega\sum_{j \in \mathcal{J}}x_{nj}s\leq P_M\label{constraint:power_constraint},\forall n\in \mathcal{N}, \\
        &\label{constraint:backhaul constraint}\quad \sum_{n \in \mathcal{N}}\sum_{k \in \Phi_{n}}\sum_{j \in \mathcal{J}}p_j(1-x_{nj})R_{nk} \leq C, \\
        &\quad s\sum_{j \in \mathcal{J}}x_{nj} \leq Q, \quad \forall n\in \mathcal{N},\label{constraint:caching capacity constraint}\\
        &\quad P_{nk} \geq 0 , \quad \forall n\in \mathcal{N}, \forall k \in \Phi_{n},\label{constraint:non-negative constraint}\\
        &\quad x_{nj}\in\{0,1\}, \quad  \forall n\in \mathcal{N},\forall j \in \mathcal{J}. \label{constraint:binary constraint}
\end{align}
\end{subequations}

Constraint (\ref{constraint:power_constraint}) makes sure that the sum of caching and transmission power does not exceed the maximum power constraint $P_M$. Constraint (\ref{constraint:backhaul constraint}) ensures that the backhaul occupancy by uncached files should satisfy the capacity constraint of fiber backhaul $C$. Constraint (\ref{constraint:caching capacity constraint}) makes sure that the total size of the files cached at ONU-AP$_n$ should not exceed its cache capacity $Q$. Constraint (\ref{constraint:non-negative constraint}) guarantees that the transmission power is non-negative. Constraint (\ref{constraint:binary constraint}) means that the cache decision on the $j$-th file at ONU-AP$_n$ is a binary variable $x_{nj}$.

\subsection{Optimal Transmission Power Allocation Method}
Problem \textbf{P1} is a typical Mixed Integer Programming (MIP) problem, which is non-linear and non-convex. However, the problem becomes a convex optimization problem if the caching placement decision $\{x_{nj} \}$ is fixed and the backhaul constraint (\ref{constraint:backhaul constraint}) is neglected, which can be written as
\begin{subequations}
\begin{align}
\textbf{P2: } \max_{\bm{P}}    &\quad\sum_{n\in \mathcal{N}}\sum_{k \in \Phi_{n}} R_{nk} \label{obj-func-P2} \\
\text{s.t.}    &\quad \sum_{k \in \Phi_{n}}\rho P_{nk}+\omega\sum_{j \in \mathcal{J}}x_{nj}^{\circ} s\leq P_M,\forall n\in \mathcal{N}, \label{power-constraint-P2} \\
&\quad P_{nk} \geq 0 , \quad \forall n\in \mathcal{N}, \forall k \in \Phi_{n}. \label{power-constraint2}
\end{align}
\end{subequations}

The objective function (\ref{obj-func-P2}) is a strictly concave and increasing function with respect to $P_{nk}$, and the inequality constraints (\ref{power-constraint-P2}) and (\ref{power-constraint2}) are convex for a given caching solution $\{x_{nj}^\circ\}$. Therefore, problem \textbf{P2} is a convex optimization problem.
By deriving the Karush-Kuhn-Tucker (KKT) conditions, we can get the optimal solution to transmission power allocation.
The Lagrangian of \textbf{P1} is given as
\begin{align}\label{Lagrangian}
    \mathcal{L}(\bm{P},\bm{\mu},\bm{\varepsilon}) = & -\sum_{n \in \mathcal{N}}\sum_{k \in \Phi_{n}}R_{nk}  -\sum_{n \in \mathcal{N}}\sum_{k \in \Phi_{n}}\varepsilon_{nk}P_{nk} \nonumber  \\
   & + \sum_{n \in \mathcal{N}}\mu_n\Big(\sum_{k \in \Phi_n}\rho P_{nk}+\omega\sum_{j \in \mathcal{J}}x_{nj} s-P_M\Big),
\end{align}
where $\bm{\mu} \in \mathbf{R}^N,\bm{\varepsilon} \in \mathbf{R}^{N\times K}$are Lagrangian multipliers.
The KKT conditions can be expressed as
\begin{subequations}
\begin{numcases}{}
\frac{\partial \mathcal{L}}{\partial P_{nk}} = - \frac{g_{nk} B }{(\sigma^{2}+g_{nk}P_{nk})\ln2}
 + \rho\mu_n - \varepsilon_{nk} = 0, \qquad \label{necessary condition} \\
\mu_n\bigg(\sum_{k \in \Phi_n}\rho P_{nk}+\omega\sum_{j \in \mathcal{J}}x_{nj}^\circ s-P_M\bigg)=0, \label{complementary slackness2}\\
\varepsilon_{nk}P_{nk}=0,\quad \forall n \in \mathcal{N}, \forall k \in \Phi_{n}, \label{complementary slackness3} \\
\mu_n,\varepsilon_{nk}\geq 0,\quad \forall n \in \mathcal{N}, \forall k \in \Phi_{n}, \label{dual feasibility}
\end{numcases}
\end{subequations}
where (\ref{necessary condition}) is a necessary condition for an optimal solution, (\ref{complementary slackness2}) and (\ref{complementary slackness3}) represent the complementary slackness, and (\ref{dual feasibility}) represents the dual feasibility.

Note that the downlink throughput in the FiWi network is maximized when the ONU-APs consume the maximum power, i.e.,satisfy
\begin{equation}
         \sum_{k \in \Phi_{n}}\rho P_{nk}+\omega\sum_{j \in \mathcal{J}}x_{nj}s= P_M
         ,\quad \forall n \in \mathcal{N}.\label{equality of power}
\end{equation}
This can be proved by contradiction. Suppose that there exists $n \in N$ satisfying $\sum\nolimits_{k \in \Phi_{n}}\rho P_{nk}$ $+\omega\sum\nolimits_{j \in \mathcal{J}}x_{nj}s< P_M$, then ONU-AP$_n$ can increase its caching power to $P_M-\sum\nolimits_{k \in \Phi_{n}}\rho P_{nk}$ for a higher cache hit ratio without reducing the sum rate of UEs associated to ONU-AP$_n$.
Another fact can be observed that files with higher popularity should be cached preferentially for maximizing the throughput in the FiWi network, which can also be easily proved by contradiction. Thus the highest-popularity-first caching strategy is adopted in this paper.

\begin{proposition}
For problem \textbf{P1}, the optimal transmission power $P_{nk}^*$ allocated to UE$_k$ by ONU-AP$_n$ can be expressed as\label{Theorem1}
\begin{equation}\label{optimal Pnk}
 P_{nk}^*=\bigg(\frac{B}{\rho\mu_n \ln2}-\frac{\sigma^2}{g_{nk}}\bigg)^+,
\end{equation}
where $(\cdot)^+=\max(\cdot,0)$, and the corresponding Lagrangian multiplier $\mu_n$ is
\begin{equation}\label{mu n without lamda}
 \mu_n=\frac{|\Phi_n|B}{\left(P_M+\sum\limits_{k \in \Phi_n}\frac{\sigma^2}{g_{nk}}-\omega\sum\limits_{j \in \mathcal{J}}x_{nj}^\circ s\right) \rho \ln2}.
\end{equation}
\end{proposition}
\begin{proof}
The closed-form expression Eq. (\ref{optimal Pnk}) can be obtained by solving the convex problem \textbf{P2}. Specifically, according to Eq. (\ref{equality of power}) and the KKT conditions mentioned above, we can obtain the expression for the optimal transmission power $P_{nk}^*$ with respect to Lagrangian multipliers $\mu_n$, as shown in Eq. (\ref{optimal Pnk}). Substituting Eq. (\ref{optimal Pnk}) into Eq. (\ref{equality of power}), we can obtain Eq. (\ref{mu n without lamda}). Note that  problem \textbf{P2} is based on the assumption that the caching placement decision $\{ x_{nj}^{\circ}  \}$ is fixed and the backhaul bandwidth constraint (5c) is not considered. To this end, what we have to prove next is that when the caching placement decision $\{ x_{nj} \}$ is not fixed and the backhaul bandwidth constraint is taken into account, Eq. (\ref{optimal Pnk}) is still the optimal transmission power allocation expression.
This can be proved by contradiction.
Suppose the optimal solution to problem \textbf{P1} is $\{P_{nk}^*,x_{nj}^*\}$ where $\{P_{nk}^*\}$ does not satisfy Eq. (\ref{optimal Pnk}). According to Eq. (\ref{equality of power}), we have $\sum\nolimits_{k \in \Phi_n}\rho P_{nk}^* +\omega\sum\nolimits_{j \in \mathcal{J}}x_{nj}^* s=P_M$. Let $\{P_{nk}^\circ,x_{nj}^*\}$ denote the solution that satisfies Eq. (\ref{optimal Pnk}), which indicates the maximum sum rate of UEs associated to ONU-AP$_n$ without the backhaul constraint. Then we have
\begin{subequations}
\begin{numcases}{}
\sum_{k \in \Phi_n}\log_2\! \Big(1+\frac{g_{nk}P_{nk}^\circ}{\sigma^2}\Big) \! > \! \sum_{k \in \Phi_n}\log_2 \! \Big(1+\frac{g_{nk}P_{nk}^*}{\sigma^2}\Big), \label{contradiction1}\qquad \\
\sum_{k \in \Phi_n}P_{nk}^\circ = \sum_{k \in \Phi_n}P_{nk}^* . \label{contradiction2}
\end{numcases}
\end{subequations}
From (\ref{contradiction1}) and (\ref{contradiction2}), we infer that by adjusting the caching and power allocation solution, there exists another set of solutions $\{P_{nk}^{\circ \circ},x_{nj}^{\circ}\}$ satisfying Eq. (\ref{optimal Pnk}) and can be expressed as
\begin{subequations}
\begin{numcases}{}
\sum_{k \in \Phi_n}\log_2 \! \Big(1+\frac{g_{nk}P_{nk}^{\circ\circ}}{\sigma^2}\Big) \! = \! \sum_{k \in \Phi_n}\log_2 \! \Big(1+\frac{g_{nk}P_{nk}^*}{\sigma^2}\Big), \qquad \label{contradiction3} \\
\sum_{k \in \Phi_n}P_{nk}^{\circ\circ} < \sum_{k \in \Phi_n}P_{nk}^* . \label{contradiction4}
\end{numcases}
\end{subequations}
With Eqs. (\ref{equality of power}) and (\ref{contradiction4}), we obtain $\sum_{j \in \mathcal{J}} x_{nj}^{\circ} > \sum_{j \in \mathcal{J}} x_{nj}^{*}$. According to (\ref{contradiction3}) and the highest-popularity-first caching strategy mentioned above, we infer that the set of solutions $\{P_{nk}^{\circ \circ},x_{nj}^{\circ}\}$ achieves a higher cache hit ratio without reducing the downlink throughput. Therefore, we get the conclusion that  $\{P_{nk}^{\circ \circ},x_{nj}^{\circ}\}$ yields the optimal solution when the caching placement decision $\{ x_{nj} \}$ is not fixed and the backhaul bandwidth constraint is taken into account, where $\{P_{nk}^{\circ \circ},x_{nj}^{\circ}\}$ satisfies Eq. (\ref{optimal Pnk}).
\end{proof}

In Eq. (\ref{optimal Pnk}), $\frac{\sigma^2}{g_{nk}}$ is the inverse of channel gain normalized by the noise variance $\sigma^2$. Eq. (\ref{optimal Pnk}) complies with the form of water-filling method \cite{Iterative-water-filling}. The basic idea of water-filling method is to take advantage of different channel conditions by allocating more transmission power to UEs with better channel conditions. We consider a vessel whose bottom is formed by plotting those values of $\frac{\sigma^2}{g_{nk}}$ for each subchannel $k$. Then, we flood the vessel with water to a depth $\frac{B}{\rho \mu_n \ln2}$. Note that the volume of water is not fixed and it depends on the caching solution $\{x_{nj}^\circ\}$. For a given caching solution, the total amount of water used is then $P_M-(\omega \sum_{j \in \mathcal{J}}x_{nj}^\circ s)$, and the depth of the water at each subchannel $k$ is the optimal transmission power allocated to it. Proposition 1 guarantees that the optimal solution to transmission power allocation still holds with the backhaul constraint, so this method is called volume adjustable backhaul-constrained water-filling method (VABWF).

\subsection{Problem Reformulation and Solution}
In the first stage, we obtain the optimal transmission power allocation method based on VABWF method. In the second stage, we propose a dynamic programming algorithm to perform joint power allocation and caching optimization to maximize the downlink throughput of FiWi networks.

By using Proposition 1 and the VABWF method, we can reduce the solution space of problem \textbf{P1} greatly. If we denote the solution space of problem \textbf{P1} as $\mathcal{A}$, then
\begin{align}\label{definition of set A}
   \mathcal{A} =  \bigg\{ \left\{P_{nk},x_{nj}\right\} \bigg | & P_{nk}=\bigg(\frac{B}{\rho\mu_n \ln2}-\frac{\sigma^2}{g_{nk}}\bigg)^+; \nonumber \\
  &  x_{nj} \geq x_{n(j+1)},j \in \mathcal{J} \bigg\}.
\end{align}
An element of $\mathcal{A}$, denoted as $A_{nj} = \{P_{nk}^\circ,x_{nj}^\circ \} \in \mathcal{A}$, should satisfy
\begin{equation}\label{set A satisfication}
  \sum_{j \in \mathcal{J}}x_{nj}^\circ = j.
\end{equation}

Let $\mu(A_{nj})$ denote whether $A_{nj}$ is chosen to be the solution, i.e.,
\begin{numcases}{\mu(A_{nj})=}
\begin{aligned}
&1, \ A_{nj} \text{ is chosen to be the solution},   \\
&0, \text{ otherwise}. \label{mu Anj definition}
\end{aligned}
\end{numcases}
Let $\varphi(A_{nj})$ denote the backhaul bandwidth occupied by ONU-AP$_n$ with respect to solution $A_{nj}$, which is defined as
\begin{numcases}{\varphi(A_{nj})=}
\begin{aligned}
&\sum_{k \in \Phi_n} p_j (1 - x_{nj})  R_{nk} \big|_{A_{nj}}, \mu(A_{nj})=1, \qquad \\
&\ 0 \  \qquad \qquad \qquad \qquad \quad ,  \mu(A_{nj})=0.
\label{omega Anj definition}
\end{aligned}
\end{numcases}

Likewise, let $\nu(A_{nj})$ denote the sum rate of UEs associated to ONU-AP$_n$ with respect to solution $A_{nj}$,which is defined as
\begin{numcases}{\nu(A_{nj})=}
\begin{aligned}
&\sum_{k \in \Phi_n}R_{nk}\big|_{A_{nj}}, \mu(A_{nj})=1,  \\
&\ 0\  \qquad \qquad \  ,  \mu(A_{nj})=0.
\label{nu Anj definition}
\end{aligned}
\end{numcases}

Then problem \textbf{P1} can be converted into the problem of determining the value of $\mu(A_{nj})$ with the aim of maximizing the downlink wireless access throughput of the FiWi network as follows,
\begin{subequations}
\begin{align}
\textbf{P3: } \max_{\mu(A_{nj})} \label{objective function2}   &\quad\sum_{n\in \mathcal{N}}\sum_{j\in \mathcal{J}}\nu(A_{nj}) \\
\text{s.t.}    &\label{constraint:backhaul constraint2}\quad \sum_{n \in \mathcal{N}}\sum_{j \in \mathcal{J}}\varphi(A_{nj}) \leq C, \\
        &\label{constraint:multiple choice constraint}\quad \sum_{j \in \mathcal{J}} \mu(A_{nj}) \leq 1 ,\quad \forall n\in \mathcal{N},\forall j \in \mathcal{J}, \\
        & \ \quad \mu(A_{nj}) \in \{0,1\}, \quad \forall n\in \mathcal{N},\forall j \in \mathcal{J}.
\end{align}
\end{subequations}

Problem \textbf{P3} is in the form of a multiple-choice knapsack problem (MCKP)\cite{Multiple-choice-knapsack}. Specifically, $N$ ONU-APs are considered as $N$ classes, while the joint caching and power allocation solution $\mathcal{A}_{nj}, j \in \mathcal{J}$ belonging to each class $n$ are considered as $J$ items. The problem is to choose no more than one item from each class such that the profit sum is maximized without exceeding the capacity $C$ in the corresponding backhaul limitation. Under the assumption that the values of $\varphi(A_{nj})$ and $\nu(A_{nj})$ are integers, we design an efficient algorithm through dynamic programming based on VABWF method as outlined in Algorithm \ref{alg2:Optimal Algorithm}.

\renewcommand{\algorithmicrequire}{\textbf{Input:}}
\renewcommand{\algorithmicensure}{\textbf{Output:}}

\normalem  

\begin{algorithm}
\caption{\mbox{Dynamic Programming Algorithm for \textbf{P3}}}
\label{alg2:Optimal Algorithm}
\LinesNumbered
\KwIn {$~\mathcal{N}$,$~\mathcal{J}$,$~\mathcal{K}$,$~B$,$~C$,$~P_M$,$~Q$,$~g_{nk}$,$~\omega$,$~s$,$~\rho$,$~\sigma$,$~\delta$;}

\KwOut {\mbox{$\{P_{nk}^\ast,x_{nj}^\ast\},$ and the downlink throughput;}}

\ForEach {$ n\in \mathcal{N}, j\in \mathcal{J}$}
{
    Calculate $\mu_n$ according to Eq. (\ref{mu n without lamda});

    $P_{nk} \Leftarrow (\frac{B}{\rho\mu_n\ln2}-\frac{\sigma^2}{g_{nk}})^{+}$;

  \mbox{Calculate $\varphi(A_{nj}), \nu(A_{nj})$ according to Eqs. (\ref{omega Anj definition})} \\ and (\ref{nu Anj definition});
}

Calculate $R(1,c), c=0,1,2,\cdots,C$;

\For {$n=2;n\leq N;n++$}
{
\ForEach {$c\in\big \{0,1,2,\cdots,C \big\}$}
{
\eIf {$c-\min\big \{\varphi(A_{nj})|\forall j \in \mathcal{J}\big\}<0$}
 {$R(n,c)\Leftarrow R(n-1,c)$;}
 {Update $R(n,c)$ according to Eq. (\ref{recursive formula});}

\If {$R(n,c)\neq R(n-1,c)$}
 {Record the index $j$ that satisfies $\zeta(n,c) \Leftarrow \arg\max_j\big\{\{R(n-1,c)\}\cup\{R(n-1,c-\varphi(A_{nj}))+\nu(A_{nj})\}\big\}$;}
}
}
\For {$n=N;n\geq1;n--$}
{
\If {$R(n,c)\neq R(n-1,c)$}
{The optimal item index for the $n$-th class in the case of backhaul bandwidth $c$ is obtained by $\zeta(n,c)$, then let $c \Leftarrow c-\varphi(A_{nj})$;}
}

Find the solution $A_{nj}$ according to  Eq. (\ref{definition of set A}) in the reverse order of n;

\Return {$\{P_{nk}^\ast,x_{nj}^\ast\}$}  and the downlink throughput;
\end{algorithm}

In Algorithm 1, define $R(n,c)$ to be the maximum downlink throughput of the FiWi network where there exist only the first $n$ classes with backhaul limitation $c$. Then we can consider an additional class to calculate the corresponding maximum throughput and the following recursive formula describes how the iterative method is performed
\begin{align}\label{recursive formula}
  R(n,c) = & \max \bigg\{ \big\{R(n-1,c-\varphi(A_{nj}))+ \nu(A_{nj}) \big\} \cup \nonumber \\
   &  \big\{R(n-1,c)\big\}     \Big| c-\varphi(A_{nj}) \geq 0, \forall j \in \mathcal{J}  \bigg\}.
\end{align}

Algorithm 1 adopts the concept of dynamic programming by progressing from smaller subproblems to larger subproblems, i.e., in a bottom-up manner. Specifically, steps 1-5 compute the backhaul bandwidth occupancy and the sum rate of UEs associated to ONU-AP$_n$  with respect to solution $A_{nj}$. Step 6 is the initialization of the dynamic programming algorithm, which obtains the solution to the smallest subproblem with only one ONU-AP, i.e., $n = 1$. Steps 7-14 are the dominated part of running time, which adopts dynamic programming to iteratively solve the subproblems from $n = 2$ to $n = N$. The reason we adopt dynamic programming is that problem \textbf{P3} has the property of an optimal substructure \cite{Dynamic-programming}, which means that an optimal solution to the problem is built from the optimal solutions of smaller problems having the same structure as the original.
Note that the constraint (\ref{constraint:multiple choice constraint}) is satisfied by placing the recursive formula in the innermost loop in step 14. In each iteration, we choose the optimum solution to the given number of classes $n$ and bandwidth limitation $c$.
Steps 15-18 obtain the joint caching and power allocation solution in a backtracking method.

For time complexity, the running time of Algorithm \ref{alg2:Optimal Algorithm} is dominated by steps 7-14. The $N$ ONU-APs in the first for-loop lead to $N$ subproblems, and the fiber backhaul capacity $C$ in the second for-loop leads to $C$ subproblems. In step 14, at most $J$ files would be traversed to compute a new solution of a subproblem. Thus the overall time complexity of Algorithm 1 is $O(NCJ)$.

\textit{Discussion}: If UEs are mobile, due to the UEs' communication handover between different ONU-APs, the number of UEs served by an ONU-AP may change. Thus, the joint caching and power allocation optimization algorithm should be rerun to update the corresponding solution when UEs' communication handover occurs. Specifically, by using the proposed VABWF method at the first stage, the expression of optimal transmission power $P_{nk}^{*}$ allocated to UE$_k$ by ONU-AP$_n$ is obtained. Since the UEs' communication handover might change the caching placement solution, the dynamic programming algorithm at the second stage should be rerun for the cache update. Note that the cache update operation can be performed when there exists idle backhaul bandwidth. When the backhaul bandwidth is fully occupied due to the UEs' requests for uncached files, then the first stage of Algorithm 1 can be executed based on a fixed caching placement decision $\{ x_{nj}^{\circ} \}$, thus obtaining a feasible power allocation solution. The dynamic programming algorithm at the second stage of Algorithm 1 will be executed when the backhaul bandwidth is available for the cache update operation.

\section{Throughput Upper Bound Analysis}
To investigate the capacity of the FiWi network with mmWave, we attempt to obtain the downlink throughput upper bound at ONU-APs. Such an upper bound can also be used to validate the effectiveness of the proposed algorithm in Section \ref{Section:joint-power-allocation-and-caching}.

\subsection{Average rate of mmWave Links}
Throughout the analysis in this section, the spatial distribution of UEs is assumed to follow independent Poisson point process (PPP) of density $\lambda$. This assumption has gained recognition due to its tractability \cite{Energy-efficiency-of, Analysis-on-cache-enabled-wireless, D2D-aware-device}, which enables us to utilize tools from stochastic geometry to quantitatively analyze the average performance of the network and is helpful to derive the upper bound of the downlink throughput in FiWi networks. Before obtaining the throughput upper bound, we firstly analyze and derive the downlink average rate of mmWave links.
Since each UE is associated with the closest AP, the probability density function (pdf) of $r$ can be given by $f_r(r) = \frac{2 r}{D^2}$.

Under the assumption of knowledge of statistical CSI, the average rate of mmWave links is the expectation with respect to the random variable $r$ and $h$, which can be written as \cite{Analysis-on-cache-enabled-wireless}
\begin{align}\label{tauPnk}
 \tau & = \mathbb{E}_{r} \left[ \mathbb{E}_{\Upsilon_{nk}} [ \log_2(1 + \Upsilon_{nk}) ] \right] \nonumber \\
   & = \int_{0}^{D} \mathbb{E}_{h}[\log_2(1 + \Upsilon_{nk})] f_r(r) \mathrm{d}r \nonumber \\
   & \overset{(a)}{=} \int_{0}^{D} \!\!\! \int_{0}^{\infty} \! \mathbb{P}\bigg[\log_2(1 + \frac{P_{nk} h G r^{-\alpha_m}}{\sigma^2})>t\bigg] \mathrm{d}t f_r(r) \mathrm{d}r,
\end{align}
where (a) holds because $\Upsilon_{nk}$ is a strictly positive random variable.

\begin{proposition}\label{proposition-ergodic-rate}
The average rate of mmWave links is given by
\begin{align}\label{tauPT1}
 & \tau(P_n^T,\lambda, D) \nonumber \\
  = & \int_{0}^{D} \! \int_{0}^{\infty} \!\! \sum_{i \in \{ \mathrm{L,N} \}}\sum_{m = 1}^{N_i} \gamma(m, r)  \mathrm{e}^{- \psi(P_n^T,\lambda, D)} f_r(r) \mathrm{d}t \mathrm{d}r,
\end{align}
where $\gamma(m, r) = (-1)^{m+1} \binom{N_i}{m} \rho_{i}(r)$, $\psi(P_n^T,\lambda, D) = \frac{m \eta_i r^{\alpha_i} (2^t + G - 1)  } {U_i + V P_n^T}$, $U_i = \frac{G N_i (D/2)^{\alpha_i} } {N_i - 1}$, $V = \frac{G} {\sigma^2 \lambda \pi D^2}$, $\eta_i = N_i (N_i!)^{-1 / N_i }$, $P_n^T$ denotes the total transmission power consumption at ONU-AP$_n$.
\end{proposition}
\begin{proof}
Please refer to Appendix \ref{proof-of-Proposition-2}.
\end{proof}

Since the average rate of UEs is taken over both the spatial PPP and the fading distribution, and we do not assume specific values for the fading coefficient and the path-loss exponent, Eq. (\ref{tauPT1}) is in an integral form. Nevertheless, it can be efficiently computed by using numerical integration toolbox as opposed to the usual Monte Carlo methods that rely on repeated random sampling to compute the results.
Besides, we take into account the impact of transmission power allocation on the network throughput and assign the optimal transmission power levels to the users, which maximizes the average rate of mmWave links. Note that the variable $P_n^T$ does not refer to the transmission power assigned to a typical user, but to the total transmission power consumption at ONU-AP$_n$. If the total transmission power consumption is determined, according to Eq. (\ref{equality of power}), the caching power consumption can be calculated. Then, the cache hit ratio can be obtained, thereby obtaining the backhaul bandwidth occupied by the UEs associated to ONU-AP$_n$, and this provides the theoretical guidance for trading the power allocated to caching for that assigned to transmission.

\subsection{Theoretical Downlink Throughput Upper Bound}
In this subsection, we obtain a theoretical upper bound of the throughput in FiWi networks by exploiting several analytical properties related to the average rate derived above.

We first define the sum rate of UEs associated with ONU-AP$_n$ as
\begin{equation}\label{Rn}
R_n \triangleq |\Phi_n| \tau(P_n^T,\lambda,D),
\end{equation}
where $|\Phi_n|$ denotes the number of UEs associated with ONU-AP$_n$. Then, the backhaul bandwidth occupied by these $|\Phi_n|$ UEs can be written as
\begin{align}\label{Cn}
  C_n =  (1-p_n^{hit})R_n = p_n^{miss} R_n,
\end{align}
where $p_n^{hit}$, $p_n^{miss}$ denote the cache hit ratio and cache miss ratio at ONU-AP$_n$, respectively. Then $p_n^{hit}$ can be calculated as
\begin{align}\label{pnhit}
  p_n^{hit}=\sum\limits_{j=1}^{\big \lfloor \frac{P_n^{ca}}{\omega}\big \rfloor}p_j =\sum\limits_{j=1}^{\big \lfloor \frac{P_M-P_n^T}{\omega}\big \rfloor}p_j.
\end{align}
Note that $p_n^{hit}$ is a concave function of $P_n^T$. This is due to the fact that $p_j < p_{j+1}$ according to the Zipf distribution of file popularity, and files with higher popularity are cached preferentially. Therefore, the cache hit ratio will decrease faster as the transmission power increases (i.e., the caching power decreases).
\begin{lemma}
$\tau(P_n^T,\lambda,D)$ is a monotonically increasing function of $P_n^T$, and $C_n$ is a monotonically increasing function of $P_n^T$.
\end{lemma}

\begin{proof}
According to (\ref{tauPT1}), we know that $\tau(P_n^T,\lambda,D)$ is a monotonically increasing function of $P_n^T$.
With Eq. (\ref{pnhit}), we can get that $\frac{\partial p_n^{miss}}{\partial P_n^{T}}>0$. Since both positive-valued function $p_n^{miss}$ and $\tau(P_n^T,\lambda,D)$ increase with $P_n^T$, we can prove that $C_n$ is a monotonically increasing function of $P_n^T$.
\end{proof}

In order to achieve a higher average downlink rate, the transmission power allocated to UEs needs to be increased, so it is intuitive that the average downlink rate monotonously increases with the total transmission power. On the other hand, the increase of transmission power results in the decrease of caching power, which will cause more fiber backhaul bandwidth occupancy. This gives us the insight that it does not make sense to increase the transmission power when the backhaul bandwidth becomes a bottleneck of the network throughput, so increasing the transmission power does not necessarily increase the network throughput.

Next, we would like to explore the relationship between the sum rate of $|\Phi_n|$ UEs associated with ONU-AP$_n$ and the backhaul bandwidth occupied by these UEs. Let $R_n$ be regarded as a function of $C_n$, i.e., $R_n = f(C_n)$, then the concavity property of $f(C_n)$ can be found.
\begin{lemma}
For a given coverage radius $D$ of ONU-APs and density of UEs $\lambda$ in the FiWi network, $R_n$ is a concave function of $C_n$.
\end{lemma}

\begin{proof}
 Please refer to Appendix \ref{AppendixA}.
\end{proof}

The significance of the concave function is that it implies the non-linear relationship between the fiber backhaul bandwidth occupancy and the downlink throughput. Namely, the increase of the fiber backhaul bandwidth occupancy results in a slower increase of the wireless downlink throughput. Considering that all the ONU-APs share the same throughput-limited fiber backhaul, we get an insight that each ONU-AP should not occupy much more backhaul bandwidth than other ones. This is because a much higher backhaul bandwidth occupancy does not lead to a much higher increase of wireless downlink throughput. Thus, an intuitive approach would be to make all the ONU-APs occupy equal fiber backhaul bandwidth such that the wireless downlink throughput is maximized. It turns out that this approach is valid when the UEs follow a uniform spatial distribution, and this particular spatial distribution also serves as a necessary and sufficient condition to ensure that the actual throughput can reach the theoretical upper bound derived below.

\begin{proposition}  \label{proposition-upperbound}
  The downlink throughput of the FiWi network is upper-bounded by
\begin{equation} \label{upperbound}
   R^+ \! = \! \mathrm{min} (\lambda \pi D^2 N \tau (P_n^T, \lambda,D), C + p_n^{hit} \lambda \pi D^2 N \tau (P_n^T, \lambda,D)).
\end{equation}
\end{proposition}
\begin{proof}
Please refer to Appendix \ref{AppendixB}.
\end{proof}
Proposition \ref{proposition-upperbound} provides the upper bound of the system downlink throughput in cache-aided FiWi  networks considering the fiber backhaul bottleneck.
The upper bound Eq. (\ref{upperbound}) is from using the Jensen's inequality (\ref{Requal}) and the bound is tight when the number of UEs served by each ONU-AP is close to $\lambda \pi D^2$, i.e.,  $|\Phi_n| \rightarrow \lambda \pi D^2 $.  In other words, the throughput upper bound is tight in the case of a uniform spatial distribution of UEs. Particularly, the throughput upper bound is achievable when $|\Phi_n| = \lambda \pi D^2 $. In practice, since UEs are modeled as independent Poisson point processes, UEs may be located closely together in certain areas and far away in others. In this case, the condition that the equality sign holds in the Jensen's inequality cannot be satisfied, which might make the actual downlink throughput lower than the theoretical upper bound.

Note that each AP consumes the same amount of power $P_n^T$ for transmission when the throughput reaches the upper bound. In other words, the upper bound of the throughput is given by Eq. (\ref{upperbound}) as long as the transmission power $P_n^T$ for each AP is determined. However, the transmission power is not given in our problem formulation and it only needs to satisfy constraint (\ref{constraint:power_constraint}), i.e., not exceeding the maximum power $P_M$. Therefore, it is necessary to determine the amount of transmission power to maximize the downlink throughput. Based on Proposition \ref{proposition-upperbound}, we then have the following corollary.

\newtheorem*{Corollary 1}{Corollary 1}
\begin{Corollary 1}
 $R^+$ is a concave function of $P_n^T$, and the downlink throughput upper bound under the maximum power constraint is the supremum of $R^+$ with respect to $P_n^T \in [0, P_M]$, i.e.,
\begin{equation}\label{supR}
    R^*(P_{nk}, x_{nj}) \leq \sup_{P_n^T \in [0, P_M]}\{R^+\}.
\end{equation}
\end{Corollary 1}
\begin{proof}
Please refer to Appendix \ref{AppendixC}.
\end{proof}

Since the throughput upper bound is obtained, we can calculate the corresponding cache capacity utilization, which by definition can be expressed as
\begin{equation}\label{cache utilization}
    \frac{1}{\omega Q}\Big(P_M-\mathop{\arg\max}_{P_n^T \in [0,P_M]}\{R^+\}\Big).
\end{equation}
This gives an analytical result on the appropriate average cache capacity utilization when the downlink throughput reaches the theoretical upper bound, and we will also compare it with the simulation results in the next section.

\section{Numerical and Simulation Results}
In this section, we present both numerical and Monte Carlo simulation results to validate the theoretical analysis and the performance of the proposed power allocation and caching optimization algorithm under different FiWi  network scenarios. The performance is averaged over 1000 network deployments in the Monte Carlo simulations.
At the fiber backhaul, ONU-APs connect with the OLT via a 1:16 splitter. At the wireless front ends, UEs are scattered in a square area of 700 m $\times$ 700 m based on independent PPPs. The traffic is generated according to UEs' requests of files, and the backhaul traffic is generated when a requested file by a UE is not cached in the associated ONU-AP. The joint caching and power allocation solution will be updated when a new network deployment is performed.
The simulation parameters are summarized in Table \uppercase\expandafter{\romannumeral2}. To better assess the performance of the proposed algorithm, which we refer to as volume adjustable backhaul constrained wafer filling-dynamic programming (VABWF-DP) below, we compare it with three benchmark power allocation and caching algorithms as follows,
\begin{itemize}
  \item \textbf{Water Filling-Full Caching (WF-FC)}: Each ONU-AP fully uses the cache capacity to store the most popular files, and the classical water-filling method is used to allocate transmission power to the UEs.
  \item \textbf{Equal Power-Popularity First (EP-PF)}: Transmission power is equally allocated to the UEs. Meanwhile, each ONU-AP chooses the most popular files to cache, but does not necessarily use the cache capacity in full.
  \item \textbf{Water Filling-Random Caching (WF-RC)}: ONU-APs randomly choose files to cache and transmission power is also equally allocated to the UEs.
\end{itemize}

\subsection{Validation of the Theoretical Analysis and Discussion}

\begin{table}
  \centering
  \caption{SIMULATION SETTING}
  \label{simulationsetting}
  \begin{tabular}{|l|p{3.6cm}|l|}
    \hline
    \textbf{Parameters} &  \textbf{Physical meaning}                  & \textbf{Values}\\\hline
    $B$  & Subchannel bandwidth               &$10$MHz\\\hline
    $C$  &  Fiber backhaul capacity  &$15$Gbps \\\hline
    $D$  &  Coverage radius of ONU-AP  & $100$m \\\hline
    $J$ & Number of files& $1000$\\\hline
    $N$ & Number of ONU-APs                            &$16$ \\\hline
    $Q$  & Cache size of each ONU-AP  &$40$GB\\\hline
    $P_M$  &  Maximum total power for each ONU-AP & $8$W \\\hline
    $P_n^{cc}$ & Circuit power at each ONU-AP   &$2$W   \\\hline
    $s$ & Size of file& $100$MB\\\hline
    $N_{\mathrm{L}} / N_{\mathrm{N}}$  & Nakagami fading parameter for LOS and NLOS channel   & $3/2$ \\\hline
    $\alpha_\mathrm{L} / \alpha_\mathrm{N}$  & Path loss exponent of LOS and NLOS channel  & $2/4$\\\hline
    $\beta$ & Blockage density             &  $0.002$   \\ \hline
    $\delta$  &  Zipf parameter  & $0.8$ \\\hline
    $\lambda$  &   Density of UEs   &  $4 \times 10^{-4}$ /m$^2$   \\ \hline
    $\rho$ & Transmission power coefficient   & $1.2$\\\hline
    $\omega$  & Caching power efficiency   &$6.25\times 10^{-12}$W/bit  \\\hline
  \end{tabular}
\end{table}
We first validate the analysis of the throughout upper bound as well as the cache capacity utilization. Fig. \ref{fig:UBColt} shows the performance of downlink throughput for different fiber backhaul bandwidths. For fair comparison, the average number of UEs is set to be equal. It can be seen that the trend of throughput upper bound of the theoretical analysis is consistent with that of the simulation results. Specifically, when the fiber backhaul bandwidth is relatively low, the downlink throughput increases linearly with the fiber backhaul bandwidth, whereas it increases slowly with a relatively high  backhaul bandwidth. We also observe that the downlink throughput is higher when ONU-APs have smaller coverage radius $D$. Quantitatively, the simulation results of the throughput have reached an average of 94.8\% of the theoretical upper bound. For the backhaul bandwidth of 15 Gbps, the corresponding downlink throughput is about 29.3 Gbps, 36.2 Gbps and 42.3 Gbps under three different AP coverage radii, respectively. This corresponds to a gain of about 95.3\%, 141.3\%, and 182.0\% with respect to a FiWi network without caches, respectively. Fig. \ref{fig:UBCacheUtilization} shows the performance of the cache capacity utilization determined by Eq. (\ref{cache utilization}) and the simulation results for different fiber backhaul bandwidths. As can be seen, the two sets of results match each other very well for all the considered backhaul bandwidth $C$. Full cache capacity is used in both theoretical analysis and simulations when the fiber backhaul bandwidth is relatively low, and the cache capacity utilization decreases as the fiber backhaul bandwidth increases. Besides, less cache capacity is used for a larger ONU-AP coverage radius.

\begin{figure*}[ht]
  \subfigure[]{
    \label{fig:UBColt} 
    \includegraphics[width=2.3in]{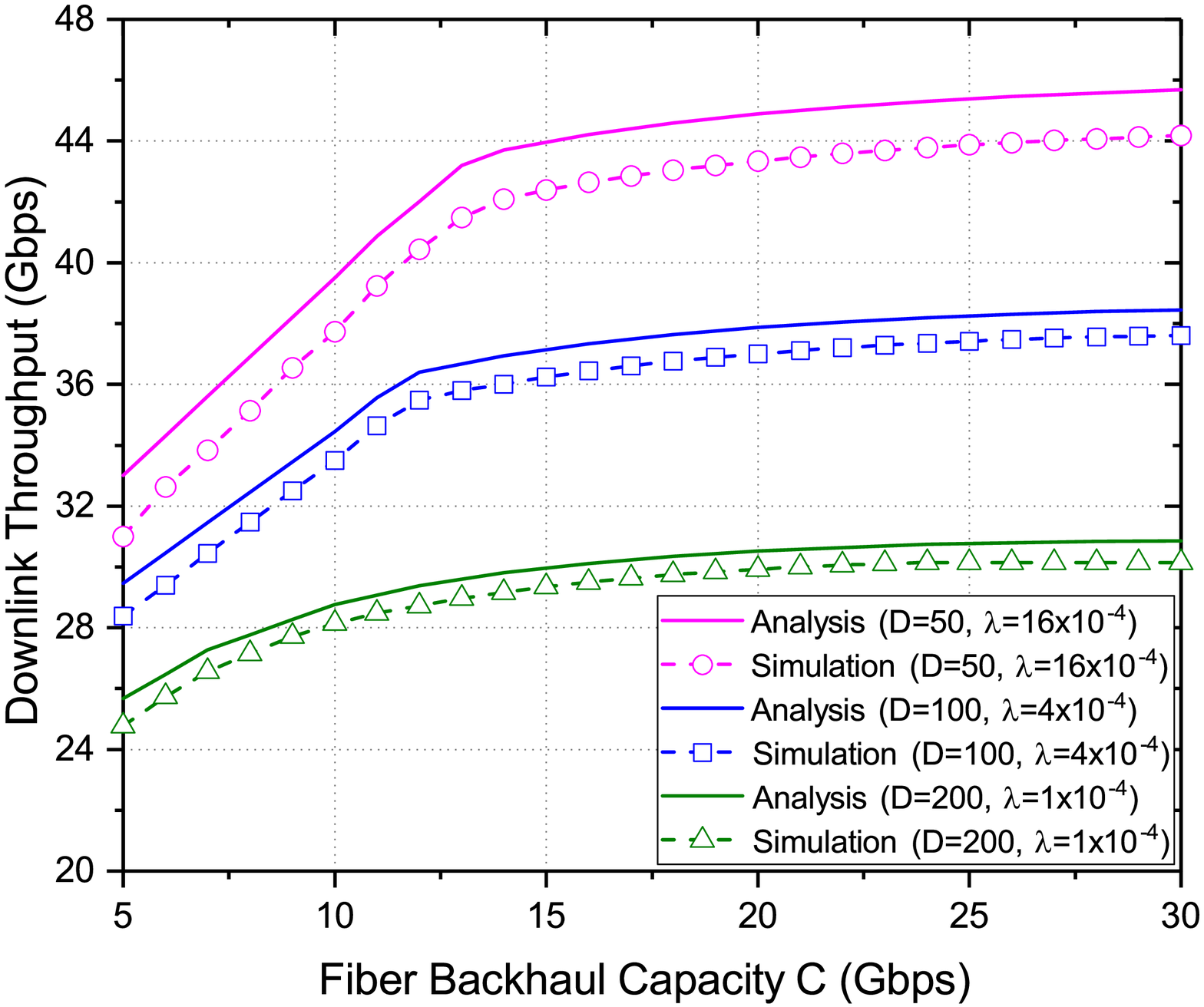}}
  \subfigure[]{
    \label{fig:UBCacheUtilization} 
    \includegraphics[width=2.3in]{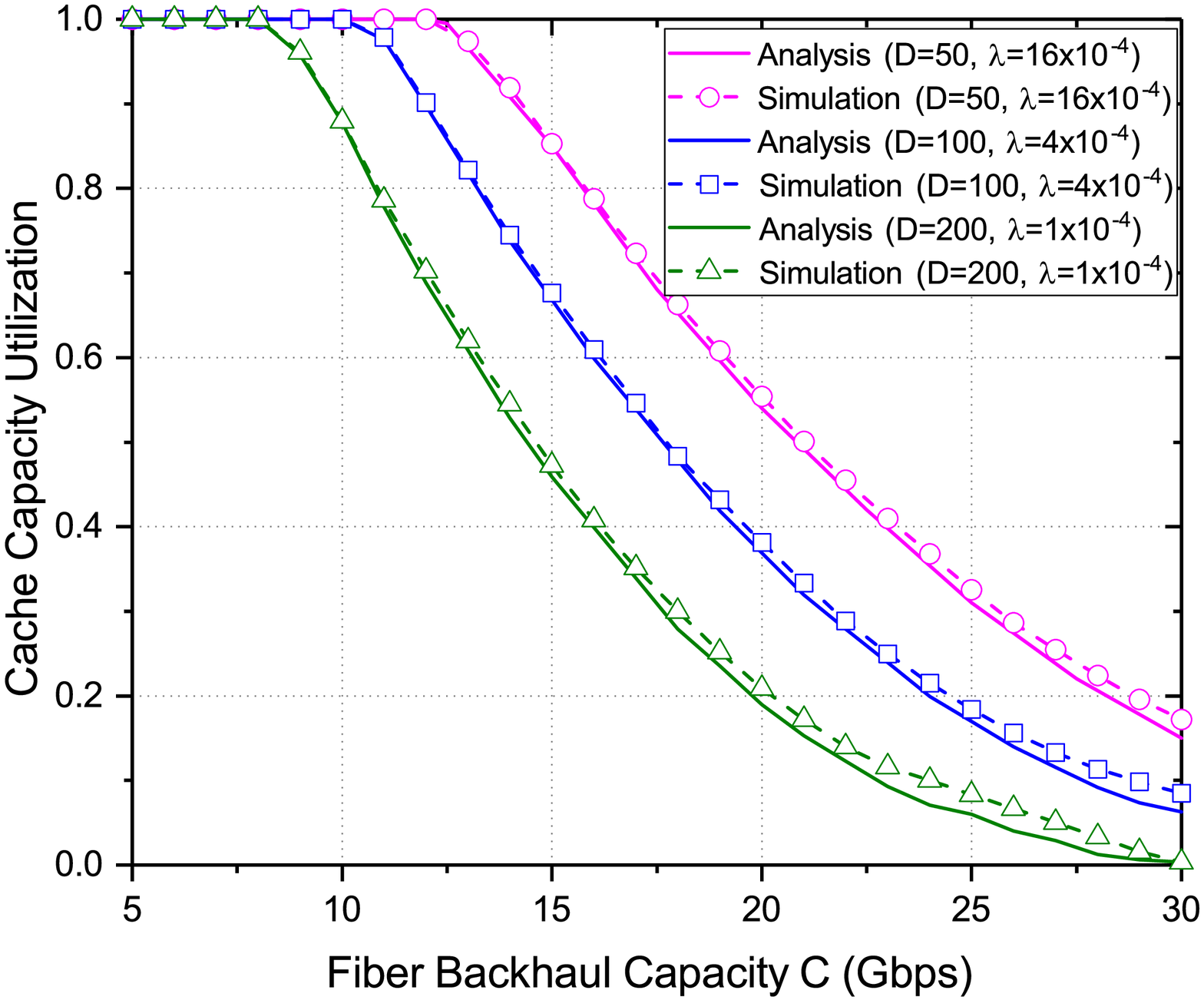}
    }
  \subfigure[]{
    \label{fig:AnalysisCacheUtilization} 
    \includegraphics[width=2.3in]{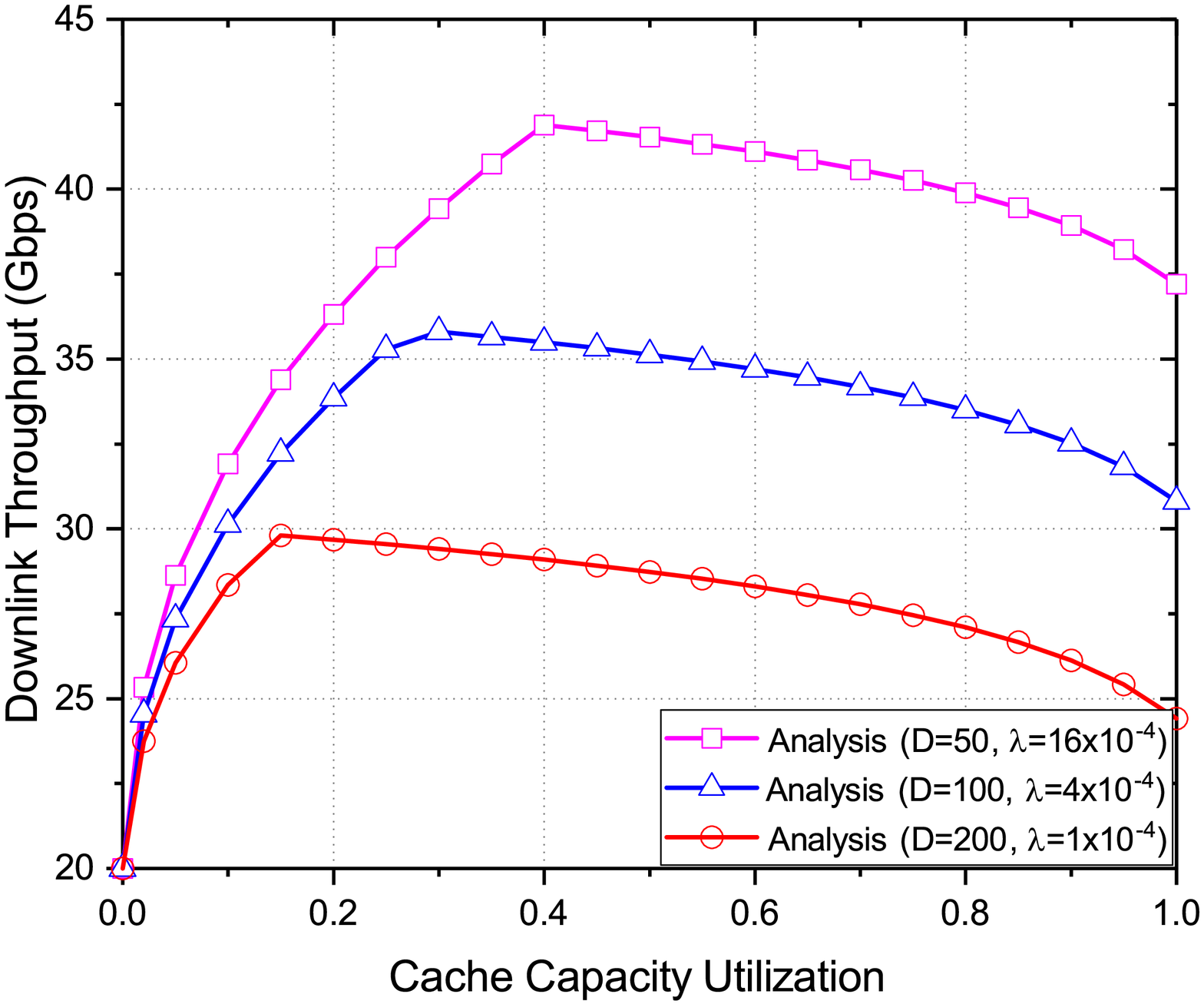}
    }
  \caption{Comparison between theoretical analysis and simulation results with respect to $C$ and $D$. (a) Throughput vs. fiber backhaul capacity; (b) Cache capacity utilization vs. fiber backhaul capacity; (c) Downlink throughput vs. Cache capacity utilization under different $D$.}
  \label{fig:2}
\end{figure*}

Comparing Fig. \ref{fig:UBColt} with Fig. \ref{fig:UBCacheUtilization}, we find that the throughput is lower when full cache capacity is used. This reflects the impact of cache capacity on downlink throughput. Due to the limited cache capacity on ONU-APs, even though full cache capacity is used, the backhaul is still unable to carry all the traffic of cache-missed files. Thus, the insufficient fiber backhaul bandwidth remains the bottleneck of the system throughput.
We find that the fiber backhaul is less likely to become a bottleneck of the system throughput when the coverage radius $D$ of ONU-APs is large. This is because the average distance between ONU-APs and UEs is farther, resulting in smaller average rate of mmWave links, and thus there is less demand for fiber backhaul bandwidth.
It can be found that the cache capacity utilization is reduced when the backhaul bandwidth is higher. This is because there is no need to use full cache capacity at this time, but rather to use more power for wireless transmission so that the wireless link capacity is increased, thereby achieving higher downlink throughput. In addition, although the UEs are spatially randomly distributed in the simulations, the proposed algorithm can approach the theoretical throughput upper bound since we have proposed an optimal transmission power allocation scheme and a dynamic programming algorithm. As will be shown later, the proposed algorithm outperforms existing algorithms in terms of the system throughput.

In Fig. \ref{fig:AnalysisCacheUtilization}, we plot the numerical results of the throughput as a function of the cache capacity utilization. It is not difficult to find that we normally do not use full cache capacity when the maximum downlink throughput is achieved, because using too much or too little cache capacity will result in a decrease in the downlink throughput. It should be noted that the cache capacity utilization here is related to the transmission power $P_n^T$ in Eq. (\ref{upperbound}). For example, if the cache capacity utilization is 0, the transmission power $P_n^T$ equals the maximum power $P_M$, in which case all the files need to be acquired through the backhaul, so the downlink throughput cannot be higher than the fiber backhaul bandwidth. Conversely, if the cache capacity utilization is 100\%, then the transmission power is minimized. In this case, it is not the backhaul bandwidth, but the mmWave link capacity becomes the bottleneck of the throughput due to the limited transmission power.
In addition, the downlink throughput is higher for smaller coverage radius $D$ of ONU-APs, and the corresponding cache capacity utilization is increased, which is consistent with the  simulation results in Fig. \ref{fig:UBCacheUtilization}. The reason for this phenomenon is that the average rate of mmWave links is higher when $D$ has a smaller value, so more caches are used to alleviate the backhaul bandwidth pressure.

\subsection{Comparison between the Proposed Algorithm and Benchmark Algorithms}
\begin{figure}[ht]
  \centering
  \subfigure[]{
    \label{fig:TransmissionPower-lambda} 
    \includegraphics[width = 2.5in]{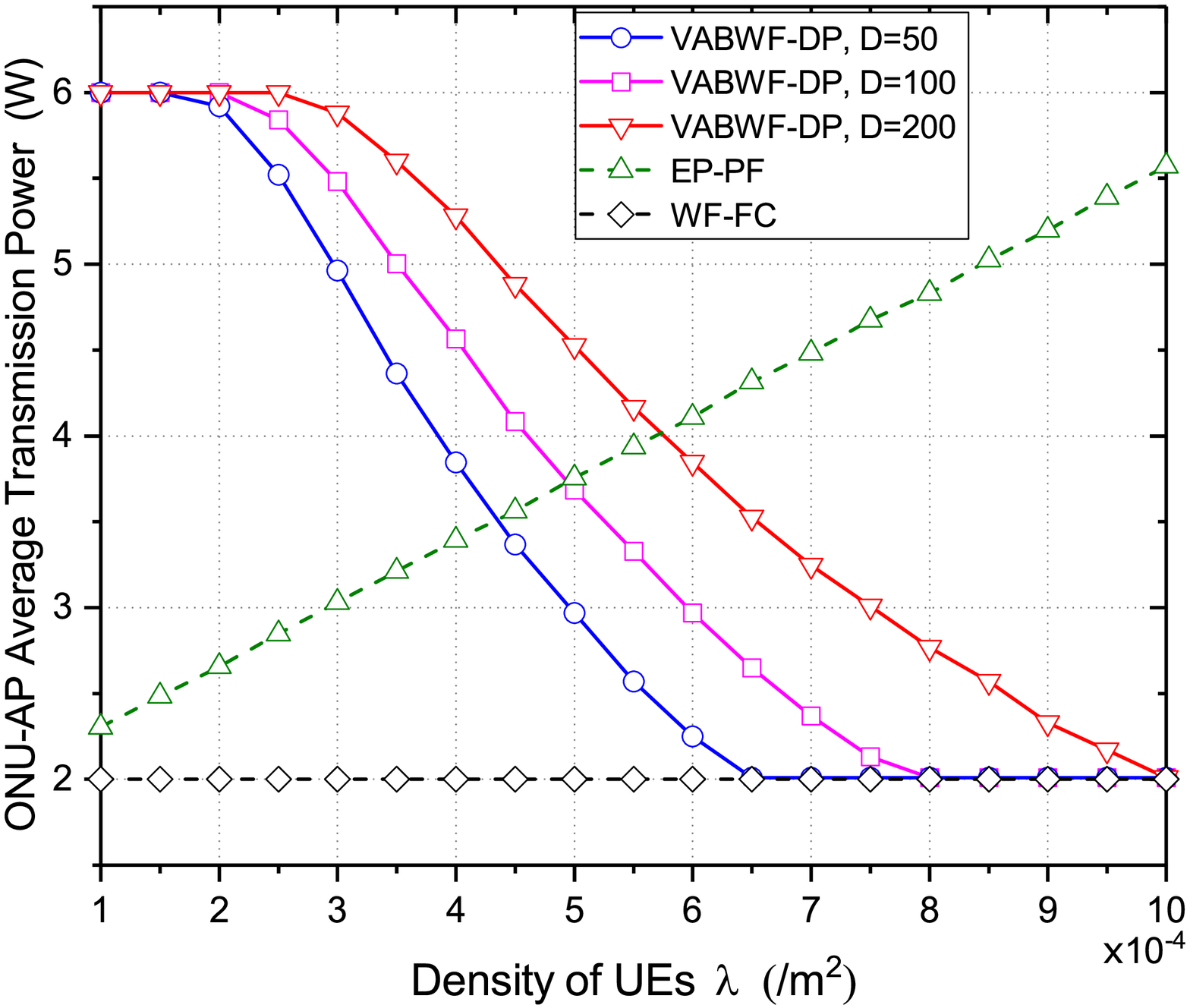}}
  \subfigure[]{
    \label{fig:Cachingprob} 
    \includegraphics[width = 2.5in]{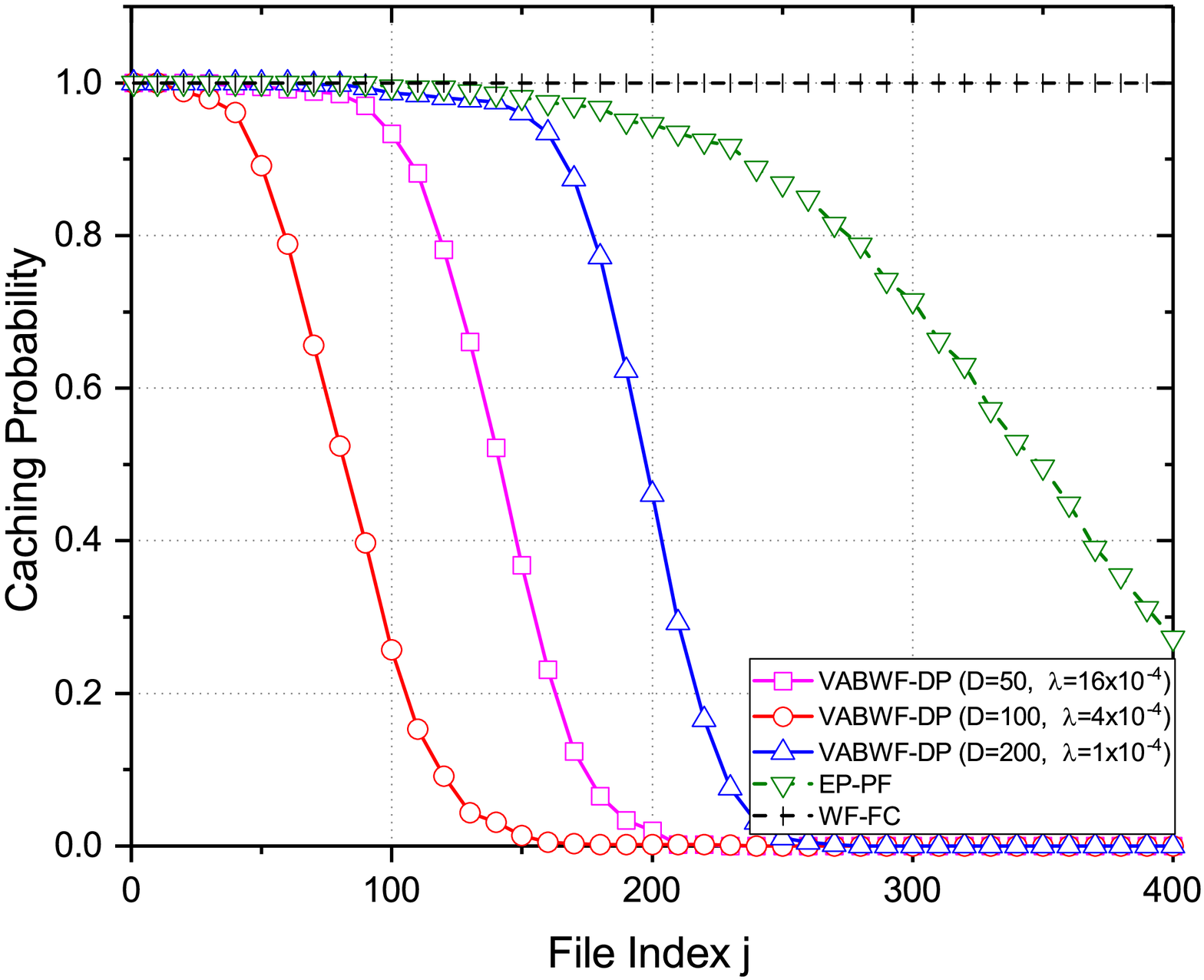}
    }
  \caption{(a) ONU-AP average transmission power for different algorithms; (b) Caching probability of files for different algorithms. }
  \label{fig:3} 
\end{figure}

Fig. \ref{fig:TransmissionPower-lambda} shows the average transmission power for different algorithms, which is defined as the average value of total transmission power $P_n^T$ consumed at ONU-AP$_n$ over Monte Carlo simulations.
It is observed that the average transmission power decreases as the number of UEs increases using the proposed algorithm VABWF-DP, while the average transmission power is unchanged by using WF-FC or increased by using EP-PF. In the case of these two benchmark algorithms, the cache hit ratio remains unchanged or even decreases as the number of UEs increases, so more requested files have to be obtained via the capacity-limited fiber backhaul, which becomes the bottleneck of the downlink wireless access throughput, thus it is wasteful to use more transmission power as the number of UEs increases. On the contrary, by using our proposed algorithm VABWF-DP, increasing the power used for caching and reducing the transmission power enable more requested files to hit the cache. Then this portion of cache-hit traffic will not be limited by the bottleneck of the backhaul bandwidth, so higher downlink wireless access throughput can be achieved.

Fig. \ref{fig:Cachingprob} shows the caching probability of files under different algorithms. In the simulation, a total number of 400 files can be cached if we use full cache capacity. By utilizing the highest-popularity-first caching strategy, files indexed after 400 will not be cached. It can be seen that VABWF-DP adjusts the caching probability of each file according to the network parameters (e.g., ONU-AP coverage radius). For instance, when the coverage radius of ONU-APs is large, most files will be cached with a lower probability, which means that the average cache capability utilization is lower. This is because the cache hit ratio does not need to be high with a relatively low mmWave link capacity, and thus more power is used for transmission to further enhance the system throughput.

\begin{figure}[ht]
  \centering
  \subfigure[]{
    \label{fig:Comparison-radius} 
    \includegraphics[width = 2.5in]{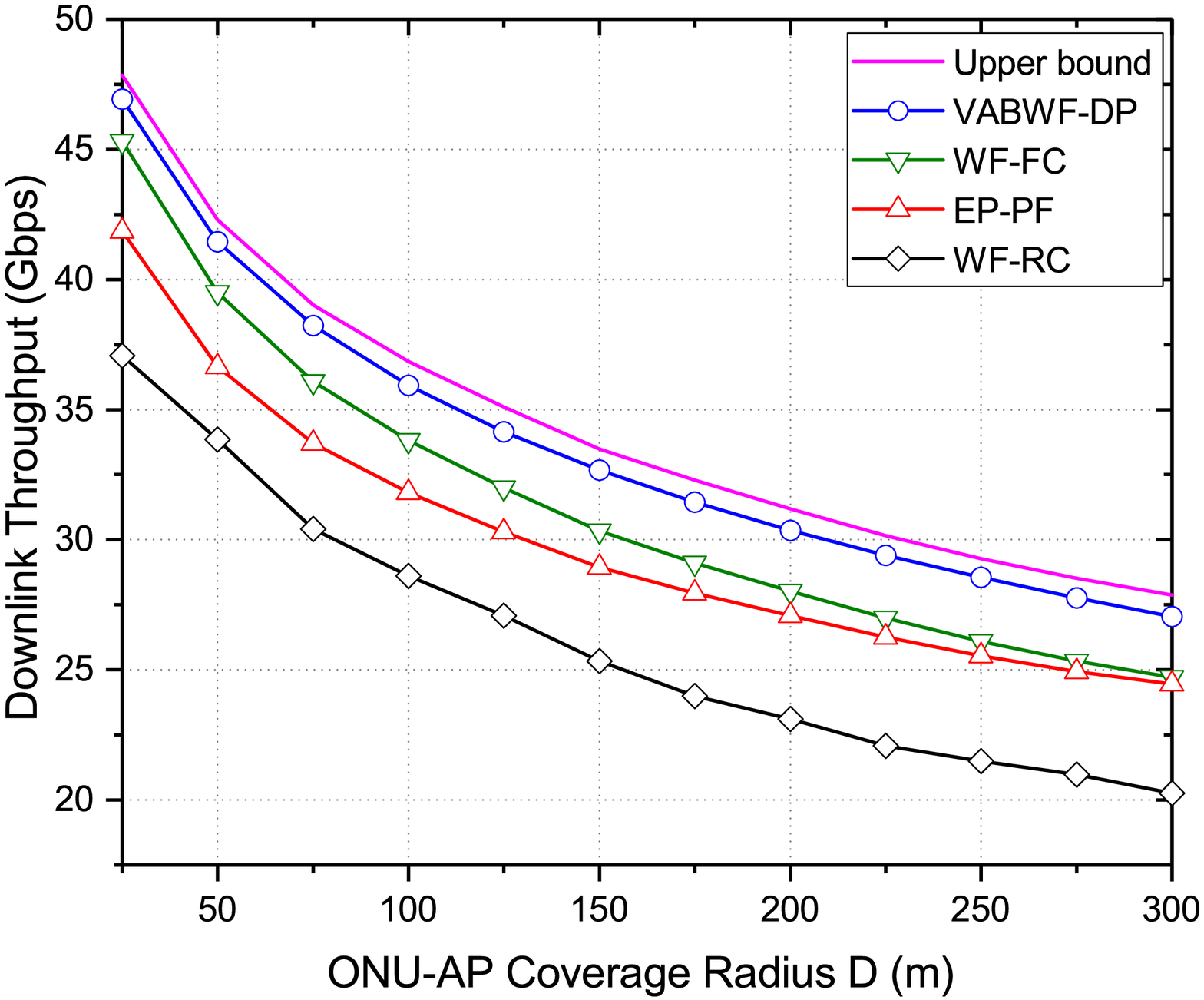}}
  \subfigure[]{
    \label{fig:Comparison-beta} 
    \includegraphics[width = 2.5in]{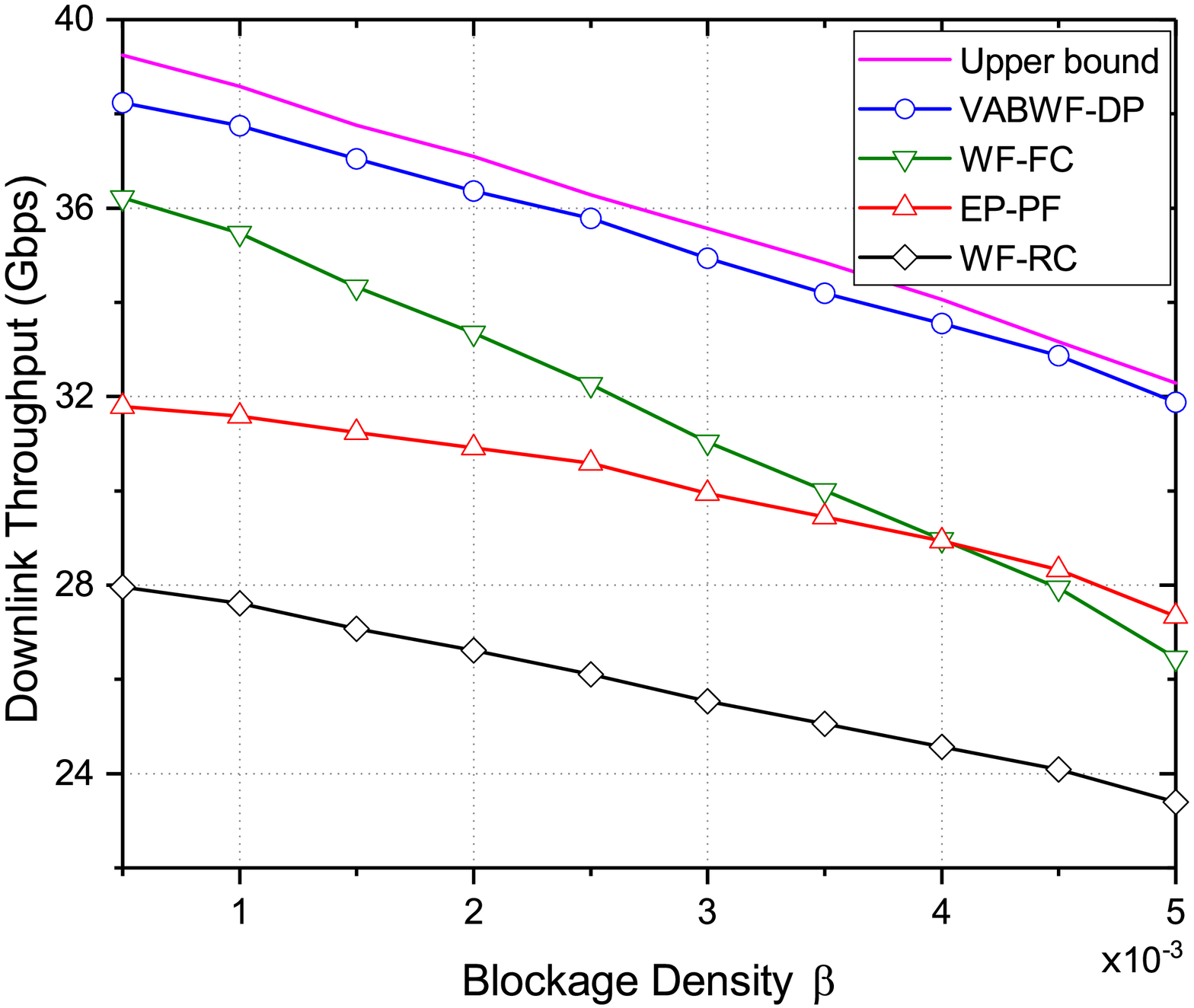}
    }
  \caption{(a) The performance of downlink throughput with respect to $D$; (b) The performance of downlink throughput with respect to $\beta$. }
  \label{fig:4} 
\end{figure}

Next, we compare the proposed algorithm with the other three algorithms in terms of the downlink throughput. Fig. \ref{fig:Comparison-radius} shows the performance of downlink throughput as the ONU-AP coverage radius $D$ varies. The average number of UEs is set to be equal. It is observed that we achieve higher throughput by using VABWF-DP at different ONU-AP coverage radii, which can be explained by Figs. \ref{fig:TransmissionPower-lambda} and \ref{fig:Cachingprob}. As the coverage radius increases, the proposed algorithm gradually increases the transmission power and reduces the number of cached files, which can provide more flexibility than the other algorithms that cannot adapt to different network conditions.

\begin{figure*}[htbp]
  \subfigure[]{
    \label{fig:Comparison-Pmax} 
    \includegraphics[width=2.28in]{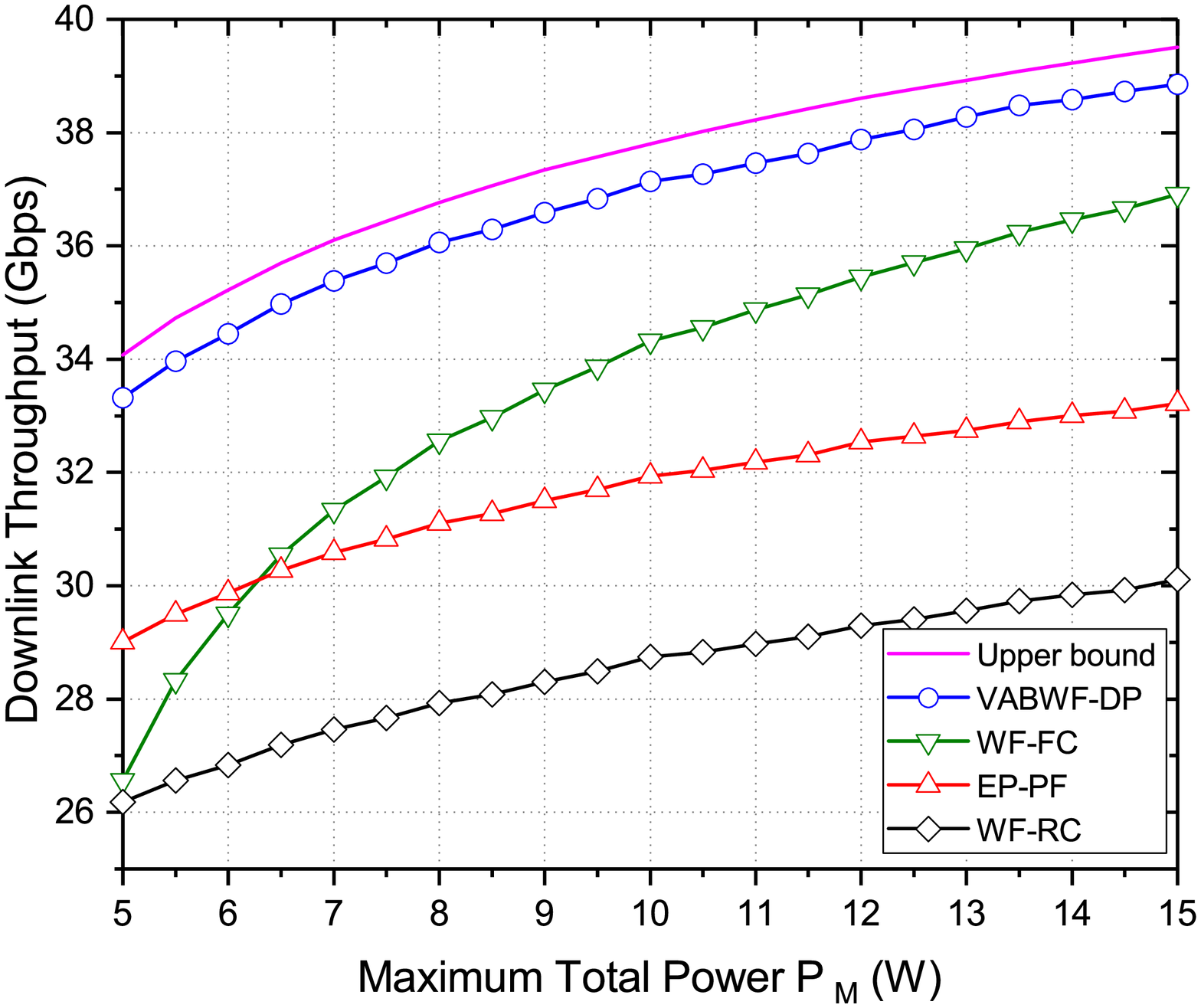}}
  \hspace{0.0in}
  \subfigure[]{
    \label{fig:Comparison-Backhaul} 
    \includegraphics[width=2.3in]{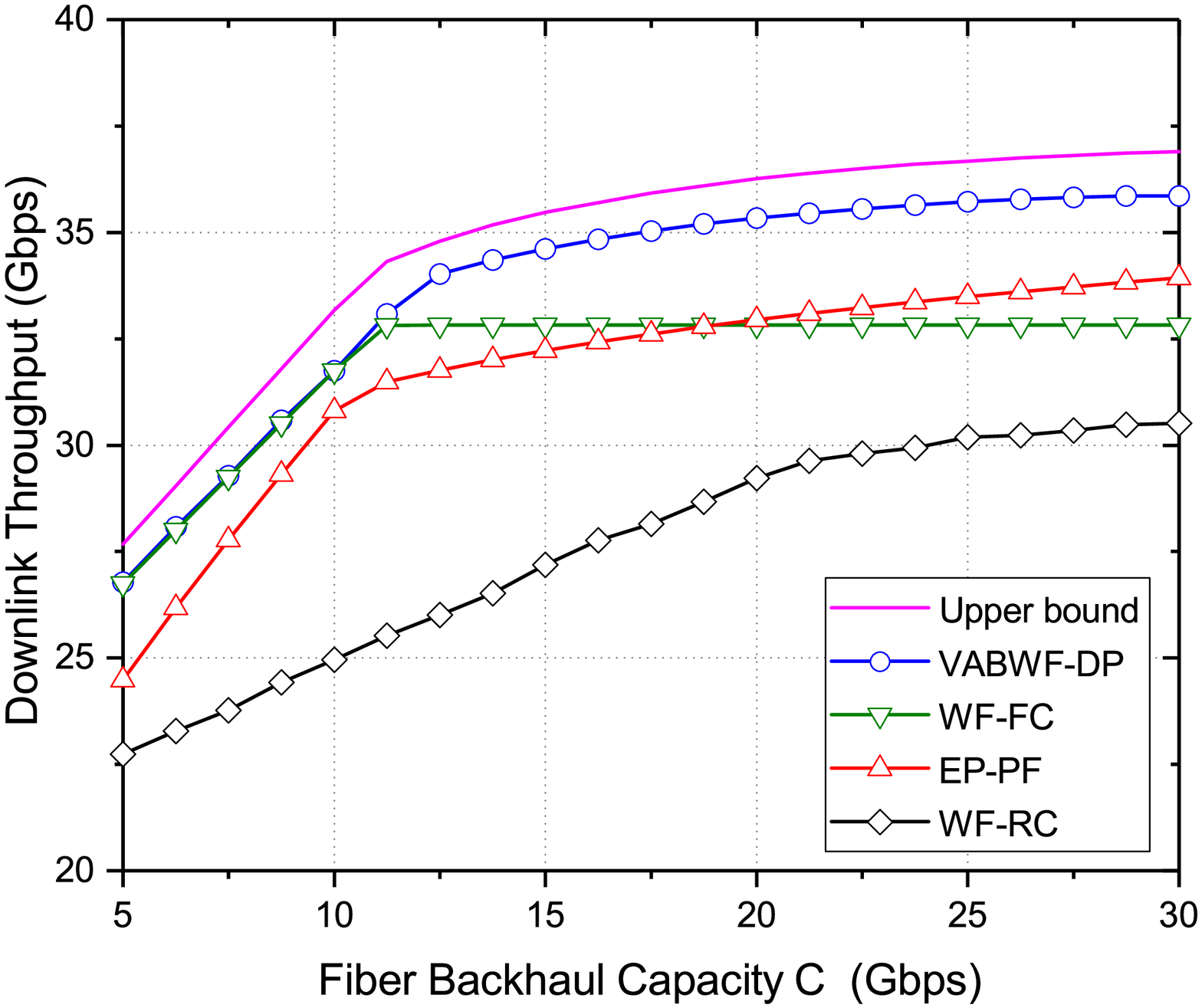}
    }
  \hspace{0.0in}
  \subfigure[]{
    \label{fig:Comparison-Zipf} 
    \includegraphics[width=2.25in]{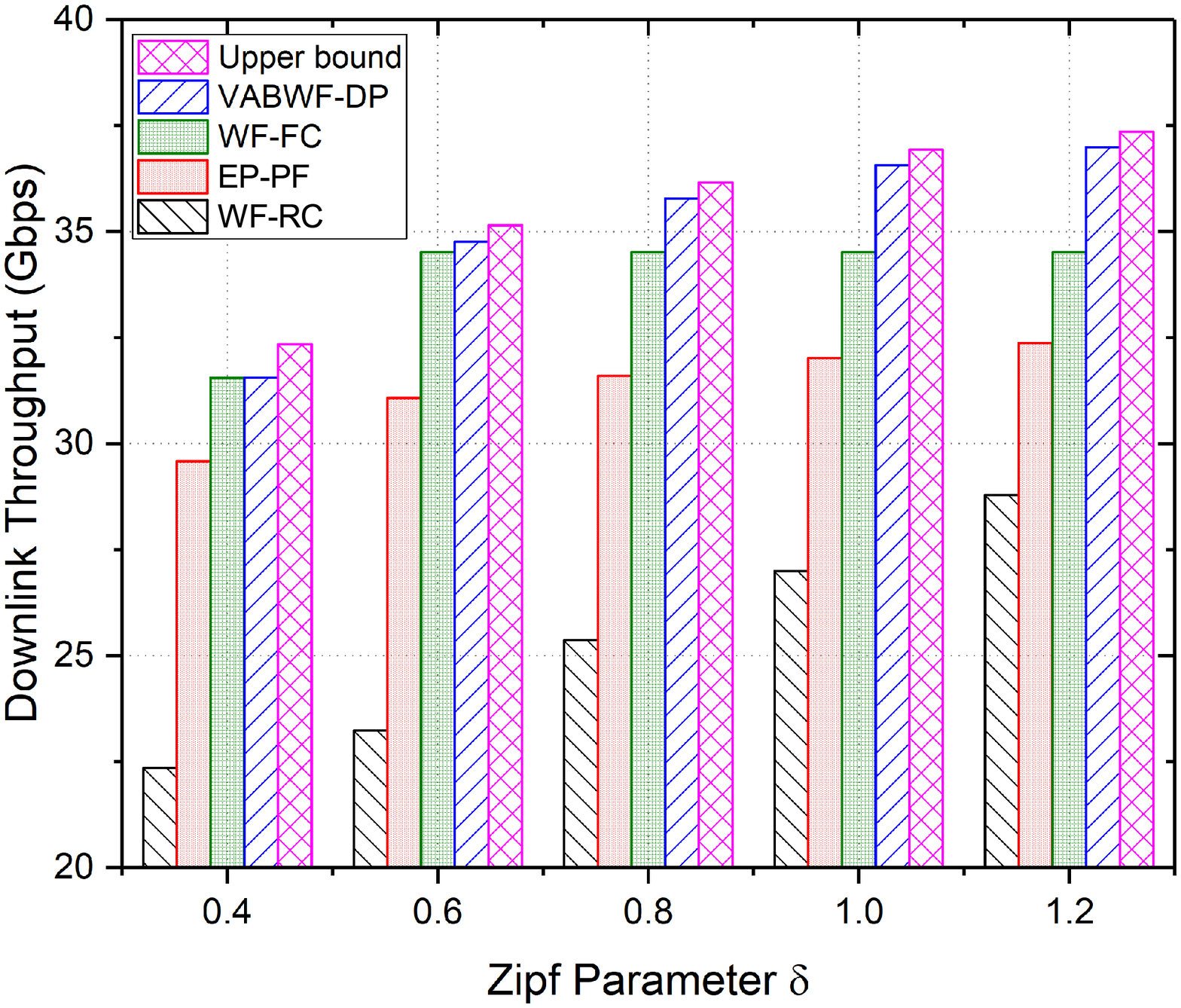}
    }
  \caption{(a) The performance of downlink throughput with respect to $P_M$; (b) The performance of downlink throughput with respect to $C$; (c) The performance of downlink throughput with respect to $\delta$.}
  \label{fig:5} 
\end{figure*}

The performance of the downlink throughput under different blockage densities is shown in Fig. \ref{fig:Comparison-beta}. The throughput decreases with the increase of blockage density, and VABWF-DP achieves higher throughput than the other algorithms. It can be observed that WF-FC achieves higher throughput than EP-PF when the blockage density is relatively low. Otherwise, WF-FC achieves lower throughput. This is due to the fact that the mmWave links are more likely to be LOS links when the blockage density is low, so the average rate of mmWave links is high and it is beneficial to cache more files to reduce the backhaul bandwidth pressure. However, with the increase of blockage density, it is better to allocate more transmission power to increase the rate of mmWave links for higher throughput. In contrast, VABWF-DP is able to adjust the power allocation and caching according to different blockage densities, so the proposed algorithm is a blockage-aware power allocation and caching method.

In Fig. \ref{fig:Comparison-Pmax}, the performance of the downlink throughput under different algorithms is shown, with the maximum total power $P_M$ of an ONU-AP ranging from 5 W to 15 W. It can be seen that VABWF-DP outperforms the other algorithms under different maximum total powers, and the gap between VABWF-DP and the upper bound of the theoretical analysis is small. This is because the proposed algorithm can adaptively adjust the transmission power and caching power of different ONU-APs.
It can be seen that the gap between VABWF-DP and WF-FC is large when the maximum total power $P_M$ is low, but the gap gradually decreases with the increase of $P_M$.
This reflects the fact that VABWF-DP and WF-FC use a more similar strategy as the maximum total power increases. In other words, by adopting VABWF-DP, more caching power is used with the purpose of mitigating backhaul bandwidth pressure with the increase of maximum total power. It is also observed that EP-PF fails to achieve higher throughput, because it does not take into account the impact of either the wireless channel conditions or the non-uniform spatial distribution of UEs. In contrast, the proposed algorithm allocates more transmission power to UEs with better channel conditions to improve the total throughput, and dynamically determines the total transmission power to cope with the throughput degradation due to the non-uniform spatial distribution of UEs. Therefore, the throughput performance of the proposed algorithm is close to the theoretical upper bound.

In Fig. \ref{fig:Comparison-Backhaul}, we evaluate the performance of different algorithms as fiber backhaul capacity $C$ varies. We can observe that the throughput performance of the proposed algorithm is still the closest to the theoretical upper bound. It is worth noting that in the case of low fiber backhaul capacity, the throughput achieved by using WF-FC is very close to that of the proposed algorithm, which is consistent with the simulation results shown in Fig. \ref{fig:UBCacheUtilization}. This observation suggests that increasing the cache capacity utilization is an effective way to achieve higher throughput when the backhaul bandwidth is low. However, WF-FC cannot keep increasing the throughput as the backhaul bandwidth increases. The reason is that the mmWave link capacity is limited by the transmission power regardless of the increase in backhaul bandwidth. In comparison, the proposed algorithm uses more power for wireless transmission instead of caching, thereby further increasing the capacity of wireless links and obtaining throughput gain.

Fig. \ref{fig:Comparison-Zipf} shows the performance of downlink throughput as the Zipf distribution parameter $\delta$ varies, which reflects the performance under various file popularities. It can be seen that the downlink throughput of the proposed algorithm increases with parameter $\delta$. This is reasonable because more user requests centralize on a smaller number of files with a larger $\delta$, while the rest of the files are rarely requested, which results in the increase of cache hit ratio and the reduction of backhaul bandwidth occupancy. Therefore, the proposed algorithm tends to cache fewer files, saving more caching power to increase transmission power and achieving higher throughput. In contrast, WF-FC fails to achieve higher throughput despite the change of Zipf parameter $\delta$. This indicates that using full-cache algorithm is rather wasteful because the files that are rarely requested are also cached in ONU-APs. Thus, no more power can be used for wireless transmission, limiting the system throughput. We also notice that EP-PF is not even comparable to WF-FC. This indicates the importance of transmission power allocation based on wireless channel conditions of UEs, which has a great impact on the system throughput.
Therefore, the proposed VABWF-DP for power allocation is more efficient in improving the throughput than EP-PF in the case of sum power constraints of ONU-APs.

\section{Conclusion}
In this paper, to cope with fiber backhaul bottleneck in FiWi networks, we equipped ONU-APs with caches to improve the system throughput. We illustrated the necessity of jointly considering power allocation and caching to enhance the network performance. We derived a closed-form expression of optimal wireless transmission power allocation using the proposed VABWF method. To perform the joint optimization of power allocation and caching, we converted the problem to a standard MCKP and designed a dynamic programming algorithm to maximize the downlink throughput. We then theoretically analyzed the average rate of mmWave links and obtained an upper bound of the downlink throughput. We proved that the theoretical throughput upper bound is achievable in the case of a uniform spatial distribution of UEs. Numerical and simulation results showed that the proposed algorithm significantly outperforms existing algorithms in terms of system throughput and is a good approximation of the derived throughput upper bound under different FiWi network scenarios. The results conclude that an appropriate allocation of transmission power and caching power is needed to achieve higher throughput in FiWi networks.

In our future work, inter-ONU-AP communications will be introduced in cache-enabled FiWi networks to better facilitate the handover issues when UEs are mobile. In this case, the cache updating may be performed through inter-ONU-AP connections when there are UEs' communication handovers between different ONU-APs, so that the cache capacity utilization at ONU-APs can be further improved, thus addressing the insufficient backhaul bandwidth issue and better facilitating cases when UEs are mobile. At the same time, resource allocation should be reconsidered since caching has a significant impact on inter-ONU-AP bandwidth allocation. Joint consideration of caching, dynamic bandwidth allocation and power allocation is a promising direction to facilitate the mobility of UEs and improve the efficiency of FiWi networks.

\appendices

\section{Proof of Proposition \ref{proposition-ergodic-rate}}  \label{proof-of-Proposition-2}
Considering that the mmWave link may be in the LOS or NLOS state, according to the law of total probability, we have
\begin{align}
& \mathbb{P}\bigg[\log_2(1 + \frac{P_{nk} h G r^{-\alpha_m}}{\sigma^2})>t\bigg]  \nonumber \\
= & \sum_{i \in \mathrm{L,N}} \rho_i(r)  \mathbb{P}\bigg[\log_2(1 + \frac{P_{nk} h G r^{-\alpha_i}}{\sigma^2})>t\bigg]  \nonumber \\
= & \sum_{i \in \mathrm{L,N}} \rho_i(r)  \mathbb{P}\bigg[h > \frac{(2^t - 1) \sigma^2 r^{\alpha_i}}{ P_{nk} G}\bigg].  \label{ergodic-Pnk}
\end{align}
To allocate the optimal transmission power $P_{nk}^*$ to UE$_k$, we refer to Eq. (\ref{optimal Pnk}) instead of adopting an equal transmission power allocation scheme. Particularly, the expectation of $\sum_{k \in \Phi_n} {\frac{\sigma^2}{r_{nk}^{-\alpha_i}  h_{nk}}}$ in Eq. (\ref{mu n without lamda}) is calculated as
\begin{align}
  \mathbb{E}\left[ \sum_{k \in \Phi_n} {\frac{\sigma^2}{r_{nk}^{-\alpha_i}  h_{nk}}} \right] & = \lambda \pi D^2 \left(\frac{D}{2}\right)^{\alpha_i} \sigma^2 \mathbb{E} \left[h^{-1} \right],
\end{align}
where $h^{-1}$ follows the inverse Gamma distribution, and the expectation of $h^{-1}$ can be calculated as
\begin{align}
  E[h^{-1}] & = \frac{(N_i)^{N_i}}{\Gamma(N_i)} \int_0^{\infty} h^{-N_i} \mathrm{e}^{-\frac{N_i}{h}} \mathrm{d} h   \nonumber \\
  & = \frac{\Gamma(N_i-1)}{\Gamma(N_i)} N_i = \frac{N_i}{N_i - 1}  ,  i \in \{ \mathrm{L, N} \},
\end{align}
where $\Gamma(\cdot)$ is the gamma function.
Substituting Eq. (\ref{optimal Pnk}) into Eq. (\ref{ergodic-Pnk}) and using the fact that the average number of UEs associated with ONU-AP$_n$ is $\lambda \pi D^2$, we have
\begin{align}
  & \mathbb{P}\bigg[h > \frac{(2^t - 1) \sigma^2 r^{\alpha_i}}{ P_{nk} G}\bigg]  \nonumber \\
  = \ & \mathbb{P} \left[ h > \frac{r^{\alpha_i} (2^t + G - 1)  } {U_i + V P_n^T }  \right]  \nonumber \\
  \overset{(a)} = \ & 1 - \left[  1 - \exp\left(-\frac{\eta_i r^{\alpha_i} (2^t + G - 1)  } {U_i + V P_n^T}\right)  \right]^{N_i}  \nonumber \\
  \overset{(b)} = \ & \sum_{n = 1}^{N_i} (-1)^{n+1} \binom{N_i}{n} \exp\left(-\frac{n \eta_i r^{\alpha_i} (2^t + G - 1)  } {U_i + V P_n^T}\right),
\end{align}
where $\eta_i = N_i (N_i!)^{-\frac{1} {N_i} }$, $U_i = \frac{G N_i (D/2)^{\alpha_i} } {N_i - 1}$, $V = \frac{G} {\sigma^2 \lambda \pi D^2}$. (a) follows from the tight lower bound of a Gamma random variable, and (b) follows by using Binomial theorem and the assumption that $N_i$ is an integer. After some algebraic manipulations, the desired proof is obtained.

\section{Proof of Lemma 2} \label{AppendixA}
 For all $P_n^T \in (0,P_{M})$, choose any $P_{n(1)}^T,P_{n(2)}^T$ that satisfies $P_{n(1)}^T<P_{n(2)}^T$. Since $C_n$ is a monotonically increasing function of $P^T$, we have $C_n \big(P_{n(1)}^T \big)< C_n \big(P_{n(2)}^T\big)$ and $p_n^{miss} \big(P_{n(1)}^T\big)<p_n^{miss} \big(P_{n(2)}^T\big)$, then
\begin{align}\label{Taylor}
   \frac{R_n^2-R_n^1}{C_n^2-C_n^1}  =& \ \frac{R_n^2-R_n^1}{p_n^{miss}\big(P_{n(2)}^T\big)R_n^2-p_n^{miss}\big(P_{n(1)}^T\big)R_n^1} \nonumber &\\
   <& \ \frac{R_n^2-R_n^1}{p_n^{miss}\big(P_{n(1)}^T\big)R_n^2-p_n^{miss}\big(P_{n(1)}^T\big)R_n^1} \nonumber &\\
   =& \ \frac{1}{p_n^{miss}\big(P_{n(1)}^T\big)} = f'\big(C_n\big).
\end{align}
By rearranging the terms in Eq. (\ref{Taylor}), we can obtain
\begin{equation}\label{concavefunction}
  R_n^2 < R_n^1 + f'\big(C_n\big)\big(C_n^2-C_n^1\big),
\end{equation}
which is the first-order Taylor approximation of $f(C_n)$. As it is always a global upper estimator of $f(C_n)$, we can conclude that $f(C_n)$ is a concave function.

\section{Proof of Proposition \ref{proposition-upperbound}} \label{AppendixB}
The downlink throughput can be calculated by adding up the rates of all the UEs, which can be expressed as
  \begin{equation}
    R=\sum_{n \in \mathcal{N}} R_n = \sum_{n \in \mathcal{N}} f(C_n).
  \end{equation}
  By applying Jensen's inequality, we obtain $R \leq N f\Big(\frac{\sum_{n \in \mathcal{N}}C_n}{N}\Big)$, where the equality sign holds if and only if all the ONU-APs occupy equal fiber backhaul bandwidth denoted by $\bar{C}$, i.e.,
  \begin{equation}\label{equalCn}
    C_n = \bar{C} = p_n^{miss} |\Phi_n| \tau(P_n^T,\lambda,D), n \in \mathcal{N}.
  \end{equation}
  By substituting Eq. (\ref{equalCn}) into Eq. (\ref{Rn}), and using the fact that $[p_n^{miss} (P_n^T)]^{-1}$ is concave, we obtain
  \begin{equation}\label{Requal}
   R = \bar{C} \sum_{n \in \mathcal{N}} [p_n^{miss} \left(P_n^T)\right]^{-1} \leq \bar{C} N \bigg[p_n^{miss} \bigg(\frac{\sum_{n \in \mathcal{N}} P_n^T}{N}\bigg)\bigg]^{-1},
  \end{equation}
  where the equality sign holds if and only if all the ONU-APs consume the same amount of transmission power, i.e.,
 \begin{equation}\label{PTequal}
   P_n^T=P^T, n \in \mathcal{N}.
  \end{equation}
     Combining Eqs. (\ref{PTequal}) and (\ref{equalCn}), we can obtain that $|\Phi_n| = \lambda \pi D^2, n \in \mathcal{N}$, which indicates that each ONU-AP serves the same number of UEs. Plugging Eq. (\ref{equalCn}) into Eq. (\ref{Requal}) and substituting $|\Phi_n|$ with $\lambda \pi D^2$, we can obtain
  \begin{equation}\label{Rleqsigma}
    R \leq \lambda  \pi D^2 N \tau (P_n^T, \lambda, D).
  \end{equation}
  Considering the constraint of fiber backhaul capacity, if $p_n^{miss} (P_n^T) \lambda \tau(P_n^T,\lambda, N ) \geq C$, we have
    \begin{equation}\label{Rleqsigma2}
    R \leq C + p_n^{hit}(P_n^T) \lambda \pi D^2 N \tau (P_n^T, \lambda, D).
  \end{equation}
  Taking the minimum of Eqs. (\ref{Rleqsigma}) and (\ref{Rleqsigma2}) gives the desired result.

\section{Proof of Corollary 1} \label{AppendixC}
To prove the concavity property of $R^+$, we first show that $p_n^{hit} \tau (P_n^T, \lambda, D)$ is a concave function of $P_n^T$. According to the analysis in Subsection \uppercase\expandafter{\romannumeral4}-B, $p_n^{hit}$ is a concave function of $P_n^T$, and $\tau (P_n^T, \lambda, D)$ is a concave function of $P_n^T$. What we need to prove is that the product of these two functions is still a concave function of $P_n^T$. Taking the second order derivative of $p_n^{hit} \tau (P_n^T, \lambda, D)$ with respect to $P_n^T$, we have
\begin{align}
&\frac{\partial^2 \big(p_n^{hit} \tau (P_n^T, \lambda, D)\big)} {\partial (P_n^T)^2} \nonumber \\
= \ & \frac{\partial^2 (p_n^{hit})}{\partial (P_n^T)^2}\tau (P_n^T, \lambda, D) + 2 \frac{\partial p_n^{hit}}{\partial P_n^T} \frac{\partial \tau (P_n^T, \lambda, D)}{\partial P_n^T} \nonumber  \\
&  + \frac{\partial^2 \tau (P_n^T, \lambda, D)} {\partial (P_n^T)^2} p_n^{hit}.
\end{align}
As $\frac{\partial p_n^{hit}}{\partial P_n^T} < 0$, $\frac{\partial \tau (P_n^T, \lambda, D)}{\partial P_n^T} > 0$, and $\frac{\partial^2 \tau (P_n^T, \lambda, D)} {\partial (P_n^T)^2}< 0$, $\frac{\partial^2 (p_n^{hit})}{\partial (P_n^T)^2} < 0$ due to the concavity properties, we can conclude that $\frac{\partial^2 \big(p_n^{hit} \tau (P_n^T, \lambda, D)\big)} {\partial (P_n^T)^2} >0$, so $p_n^{hit} \tau (P_n^T, \lambda, D)$ is a concave function of $P_n^T$.
Then the pointwise minimum $R^+$ of $\lambda \tau(P_n^T, \lambda, D)$ and $C + p_n^{hit} \lambda \tau (P_n^T, \lambda, D)$, defined by
\begin{equation}
R^+(P_n^T) = \mathrm{min}\left(\lambda \tau(P_n^T, \lambda, D), C + p_n^{hit} \lambda \tau (P_n^T, \lambda, D)\right),
\end{equation}
is also concave due to the pointwise minimum operations that preserve concavity of functions. In this case, there exists a unique supremum with respect to $P_n^T \in [0, P_M]$, which is the downlink throughput upper bound under the maximum power constraint.

\bibliographystyle{ieeetr}

\end{document}